\documentclass[11pt]{article}
\usepackage[letterpaper,margin=1in]{geometry}
\usepackage[hidelinks]{hyperref}
\usepackage{xcolor}
\usepackage{booktabs} 
\usepackage{natbib}
\usepackage{graphicx}
\usepackage{url}

\usepackage{amsmath,amsthm,amsfonts,nicefrac}
\usepackage{etoolbox}
\usepackage{thmtools,dsfont}
\usepackage{enumerate}
\usepackage{subcaption}
\usepackage{cleveref}

\Crefname{lemma}{Lemma}{Lemmata}

\newtheorem{theorem}{Theorem}
\newtheorem{lemma}{Lemma}

\newtheorem{proposition}{Proposition}
\newtheorem{definition}{Definition}
\theoremstyle{definition}
\newtheorem{remark}{Remark}
\newtheorem{claim}{Claim}

\newcommand{\E}[1]{\mathbb{E}\left[#1\right]}
\newcommand{\Esub}[2]{\mathbb{E}_{#1}\!\left[#2\right]}
\renewcommand{\P}[1]{\mathbb{P}\left(#1\right)}
\newcommand{\Psub}[2]{\mathbb{P}_{#1}\!\left(#2\right)}

\newcommand{\gft}{\mathrm{GFT}}
\newcommand{\app}[2]{\phi_{#1}\left(#2\right)}

\let \eps \varepsilon

\newcommand{\Var}[1]{\mathbb{V}\left[#1\right]}
\newcommand{\prof}{\textnormal{\textsf{profit}}}
\newcommand{\cS}{\mathcal S}

\newcommand{\opt}{\textnormal{OPT}}

\newcommand{\ind}[1]{\mathds 1_{\{#1\}}}
\newcommand{\util}[3]{U_{#3}\left(#1, #2\right)}

\newcommand{\M}{\mathcal M}
\newcommand{\A}{\mathcal A}

\newcommand{\clean}{\mathcal {E}}

\newcommand{\pathspace}{\mathcal{P}_\eta}
\newcommand{\neigh}[1]{\mathcal N^+_\eta(#1)}
\newcommand{\neighedge}[1]{E^+_\eta(#1)}

\newcommand{\expalg}{\mathcal A^{\textnormal{exp}}}

\newcommand{\cA}{\mathcal {A}}

\newcommand{\cD}{\mathcal {D}}
\newcommand{\cE}{\mathcal {E}}
\newcommand{\cG}{\mathcal {G}}
\newcommand{\cF}{\mathcal F}
\newcommand{\cL}{\mathcal {L}}
\newcommand{\cY}{\mathcal {Y}}

\newcommand{\cT}{\mathcal {T}}

\newcommand{\vs}{v_s}
\newcommand{\vb}{v_b}
\newcommand{\bs}{b_s}
\newcommand{\bb}{b_b}

\newcommand{\Mnet}{\hat \M}

\newcommand{\acc}{\alpha}
\newcommand{\fail}{\beta}

\title{Private Learning in Bilateral Trade}

\author{Simone Di Gregorio\thanks{Dept. of Computer, Control and Management Engineering, Sapienza University of Rome, Rome, Italy. Email: \texttt{\{digregorio,fuscof,leonardi\}@diag.uniroma1.it}}
\and Federico Fusco\footnotemark[1] \protect\\
\and Stefano Leonardi\footnotemark[1]
\and Chris Schwiegelshohn\thanks{Dept. of Computer Science, Aarhus University,
Aarhus, Denmark. Email: \texttt{schwiegelshohn@cs.au.dk}}}

\date{}
\begin{document}
\maketitle

\begin{abstract}
Bilateral trade models one of the most fundamental economic interactions: the intermediation between two strategic agents, a seller and a buyer, willing to trade a good. We consider the learning version of the problem, where the goal is to learn a mechanism from a sampled dataset of agents' valuations to maximize either profit or economic efficiency. While known learning algorithms are characterized by high sensitivity to the input dataset, we specifically study this problem through the lens of differential privacy, ensuring that each data point does not significantly affect the probability of learning any specific mechanism.

For our results, we adopt the PAC-learning framework: with high probability, the learning algorithm should output a mechanism that is at most an additive $\alpha$ away from optimal, in a $\eps$-differentially private way. As a first result, we show that differential privacy and (near)-optimality are not achievable for general distributions. Surprisingly, assuming that the distribution underlying the agents' valuations is $\sigma$-smooth, we recover nearly optimal sample-complexity bounds for both economic efficiency and profit. 

For profit, we show how to construct in polynomial time an $\alpha$-optimal and $\eps$-differentially private mechanism using $\tilde\Theta(\nicefrac{1}{\sigma\eps\alpha^2})$ samples. For efficiency, measured by the gain from trade, we achieve the same result using $\tilde\Theta(\nicefrac{1}{\eps\alpha}+\nicefrac{1}{\alpha^2})$ samples. Notably, these bounds are essentially tight in the precision parameter $\alpha$, since achieving $\alpha$-optimality (ignoring differential privacy) requires at least $\nicefrac{1}{\alpha^2}$ samples.

\end{abstract}
\pagenumbering{arabic}

\section{Introduction}

    The bilateral trade problem captures a fundamental economic primitive: how should a broker intermediate between a seller and a buyer wishing to trade a specific good? The crucial tension underlying this task arises from the competing interests of the two rational agents ---high vs. low price--- who, however, need to agree to make the trade happen. Bilateral trade has been investigated intensively in the economic literature, starting from the seminal works of \citet{Vickrey61} and \citet{MyersonS83}, and has attracted a considerable attention by computer scientists \citep[e.g.,][]{Colini-Baldeschi16,BrustleCWZ17,BabaioffGG20,BlumrosenD21,HajiaghayiHPS25,KangPV22}, mainly due to its natural applicability in online markets (e.g., Online Marketplaces, Online Advertising Markets, Ride-sharing Platforms).

    The seller and the buyer are characterized by their private valuations for the good, which are drawn from some underlying distribution, and the goal of the broker is to design a ``good'' mechanism, which is robust to the strategic manipulations of the two agents. In this paper, we investigate mechanisms that are  dominant-strategy incentive compatible (DSIC), ensuring agents optimize their utility by reporting their true valuations, and individual rational (IR), guaranteeing that participating in the mechanism yields non-negative utility. To assess the quality of the outcome of the mechanism, we consider the two most natural quantities: 
    \begin{itemize}
        \item \textbf{Economic Efficiency}, measured in terms of the increase in social welfare induced by the trade (the so-called gain-from-trade).
        \item \textbf{Profit}, which is simply the net utility that the broker extracts from the agents. 
    \end{itemize}
    When dealing with maximizing gain from trade, we require the mechanism to enforce a third property: budget balance (BB), so that the payment made to the seller does not exceed the payment collected from the buyer. Budget balance is needed to avoid the undesirable situation in which the broker \emph{subsidizes} the market. The addition of this constraint has the surprising effect of \emph{collapsing} the family of DSIC, IR and BB mechanisms for gain-from-trade maximization to the simple sub-family of fixed-price mechanisms, where the broker posts one price to both agents and the trade happens if they both accept \citep[][see also \Cref{prop:fixed_price}]{Hagerty87}. 

    A fruitful line of work has investigated the quality of such mechanisms under the assumption that the distribution underlying agents' valuation \emph{is known to the broker}, trying to assess the gap between the best mechanism and the best possible efficiency/profit achievable \citep[e.g.,][]{BlumrosenD21,DengMSW22}, between simple and optimal mechanisms \citep[e.g.,][]{BrustleCWZ17,CaiW23}, or the trade-off between the two objectives \citep{HajiaghayiHPS25}. Recently, a keen interest has focused on lifting the assumption of perfect knowledge of the agents' distribution. Different perspectives have been adopted along this direction: the ``single-sample'' regime \citep{DuttingFLLR26}, the (online) learning paradigm \citep{AzarFF24,Cesa-BianchiCCF24jmlr,BernasconiCCF24,LunghiCM26}, and the investigation of learning dynamics \citep{DengMSWW25}.

    A critical challenge that remains unaddressed 
    is protecting the privacy of the agents whose data is utilized to train the mechanism. Neglecting this aspect poses a significant risk, as the optimized mechanism itself may leak sensitive information about the valuations contained in the training data set. Indeed, the state-of-the-art learning algorithms for bilateral trade are highly data dependent, so that it is easy to recover some of the points used to train them \citep{DiGregorioDFS25}. Moreover, this sensitivity makes them vulnerable to adversarial manipulation of the underlying dataset, as a single seller-buyer pair in the training set may drastically change the output mechanism. In this work, we tackle this problem through the lens of differential privacy (DP), the gold standard for algorithmic privacy \citep{DworkMNS16,DworkR14}.
    
    A possibly randomized algorithm $\cA(S)$ that learns a mechanism from a sampled dataset $S$ is said to be $\eps$-differentially private if, for any two neighboring datasets $S, S'$ differing by exactly one data point, and any possible mechanism $M$, it holds that $\mathbb{P}\left[\cA(S)=M\right] \leq e^{\varepsilon}\mathbb{P}\left[\cA(S')=M\right]$, where the probability is with respect to the internal randomness of $\cA$. This condition ensures that the contribution of any single data point has a strictly limited impact on the resulting mechanism. The primary goal of this paper is to design differentially private algorithms for learning bilateral trade mechanisms that maximize either profit or gains from trade.

    \subsection{Our Results}

        We frame this problem within the standard PAC-learning paradigm, investigating the sample complexity of achieving differential privacy and (near)-optimal gain-from-trade/profit. Given a privacy requirement and precision parameter $\alpha$, we aim at finding a mechanism that is at most $\alpha$-far from optimal, in a differentially private way, while using as few samples as possible from the agents' valuations distribution. As a first result, we prove that such private learning is unattainable on certain distributions, for both gain-from-trade (\Cref{thm:impossibility_gft}) and profit maximization (\Cref{thm:impossibility_profit}).
        Surprisingly, once we assume that the underlying distribution respects the natural $\sigma$-smooth assumption\footnote{A distribution is $\sigma$-smooth if it admits a probability density function which is uniformly bounded by $\nicefrac{1}{\sigma}$ or, in words, if it is not ``too peaked''. See \Cref{def:smoothness} for a formal definition.} \citep[e.g.,][]{HaghtalabRS24}, then we show positive results for private learning:
        \begin{itemize}
            \item For Efficiency Maximization, measured in terms of the gain from trade, we show that the sample complexity\footnote{Here and throughout the paper, $\tilde{O}$ hides terms that are poly-logarithmic in the parameters.} is $\tilde{O}(\nicefrac{1}{\alpha^2}+\nicefrac{1}{\eps\alpha})$ (\Cref{thm:main_gft}). 
            \item For Profit Maximization, we show that the sample complexity is $\tilde{O}(\nicefrac1{\eps\sigma\acc^2})$ (\Cref{thm:main_profit}).
        \end{itemize}

    \paragraph{Discussion on Tightness} The two sample complexity bounds are tight in the accuracy parameter ${\alpha}$, up to poly-log terms. Indeed, even if we forget about differential privacy and only look at the problem of learning the ``best'' mechanism for bilateral trade, it is not possible to be $\alpha$-optimal using less than $\nicefrac{1}{\alpha^2}$ samples (Theorem 2 of \citet{Cesa-BianchiCCF24} for gain from trade and Theorem 5 of \citet{DiGregorioDFS25} for profit maximization). As it is common in the differential privacy literature, the privacy term $\eps$ is to be considered a small constant \citep[see, e.g., Remark 1 of][]{Dwork08}, so that the sample complexity for gain from trade is essentially optimal. The bound for profit exhibits an additional $\nicefrac 1{\sigma}$ factor, which is a distribution-dependent parameter. 

    \paragraph{Implementability} Our two positive results are algorithmic: we show how to instantiate the exponential mechanism \citep{McSherryT07} on suitable finite classes of mechanisms so as to obtain the desired properties. Notably, the support of the exponential mechanism used for the profit maximization task is \emph{exponential} in $\nicefrac{1}{\alpha}$, but we show how to sample efficiently from it! We achieve this by taking inspiration from a sampling technique from the online learning literature \citep{TakimotoW03} used to implement efficiently the Hedge algorithm; to the best of our knowledge, we are the first to exploit that argument for the exponential mechanism, and we envision it can have a broader applicability in differential privacy. We specify that the results reported so far assume that the learner knows the smoothness parameter $\sigma$. If this is not the case, we can still implement differentially private learning algorithms with a ``precision/sample-complexity'' trade-off which is slightly unbalanced in terms of $\sigma$ (\Cref{rmk:profit,rmk:gft} for details).

    \paragraph{Uniform Convergence for Profit and Gain From Trade} As a by-product of our approach, we show uniform convergence for the task of profit-learning in bilateral trade, under the smoothness assumption, with a nearly-tight sample complexity bound of order $\tilde{\Theta}(\nicefrac{1}{\acc^2})$ (Appendix~\ref{app:uniform_convergence}). Stated differently, we prove that, given enough samples, it is possible to approximate the \emph{expected} profit of any mechanism with its \emph{empirical} one---at a very fast rate\footnote{Note, to estimate the expected value of a \emph{single} Bernoulli distribution up to precision $\alpha$ we still need $\Omega(\nicefrac{1}{\acc^2})$ samples.}. This proves a further qualitative separation between smooth and non-smooth distributions: in \citet{DiGregorioDFS25} the authors prove that uniform convergence is not possible for general distributions. Interestingly, there is no such separation for \emph{gain-from-trade} maximization: in \citet{Cesa-BianchiCCF24}, the authors show fast uniform convergence bounds for the functions of the form $p \to \Esub{v}{\gft(M_p,v)}$, for general distributions. Nevertheless, private learning is unattainable also for fixed-price mechanisms.
    
    \subsection{Technical Challenges}

    We start by reviewing the state-of-the-art learning algorithms for bilateral trade and observe that they are highly ``non-private'', so that we need fresh ideas to attain private learning. 
    
    For gain-from-trade maximization, it is possible to get an $\alpha$-optimal posted-price mechanism by running empirical risk minimization on $\tilde{\Theta}(\nicefrac{1}{\alpha^2})$ samples (or follow-the-leader, in the online learning jargon), via the uniform convergence property mentioned above. Clearly such an approach is inherently data-dependent. First, if all the samples are concentrated in one point, then the output mechanism will also heavily depend on that point. Second, and most crucially for differential privacy, if there are two prices that are equally good on the sampled dataset, then changing one single valuation may change completely the identity of the price output. Moving to profit maximization, the learning algorithm Simplify-the-Best-Mechanism used in \citet{DiGregorioDFS25} is even more drastically data dependent. It works as follows: first, it computes the empirical risk minimizing mechanism $\M^{\star}$ on the sampled dataset, then, it ``simplifies'' it, by adapting the allocation region of $\M^{\star}$ to a geometric tessellation of the $[0,1]^2$ square whose tiles \emph{are induced by the observed data points}. Both these two steps are inherently non private.

    \paragraph{Private Learning Frameworks via Exponential Sampling} A popular method to obtain differentially-private algorithms is to combine the exponential mechanism \citep{McSherryT07} with a \emph{fixed} and \emph{finite} subfamily of candidate mechanisms \citep{BeimelNS19,KasiviswanathanLNRS11}. Succinctly, it entails sampling a mechanism $M$ with probability proportionate to $\exp\left(\varepsilon\cdot f(M)\right)$, where the score function $f$ denotes either the empirical profit or the gain from trade, depending on the problem at hand. This exponential sampling procedure guarantees $\varepsilon$ differential privacy, while preferring mechanisms with high empirical score. The applicability of such approach is limited by three factors. (i) The best mechanism on $\Mnet$ should be comparable with the best \emph{overall} mechanism, so that a good mechanism gets picked with some probability. (ii) The empirical scores (i.e., profit or gain from trade) are tightly concentrated around the true scores, as otherwise the distribution might output low-value mechanisms. (iii) A naive implementation requires enumerating over all candidate mechanisms $\Mnet$, which may lead to an infeasible running time. 

    \paragraph{A Fixed $\Mnet$ does not work for general distributions} Unfortunately, both the gain from trade and the profit objectives \emph{are not} approximable with fixed discretizations! This is a well-known feature of bilateral trade, for both profit \citep{DiGregorioDFS25} and gain from trade \citep[e.g.,][]{AzarFF24} which makes its learning typically challenging. We use the construction underlying this fact to formally argue that it is not possible to privately learn either objective in bilateral trade, without making some assumption on the underlying distribution. To overcome these impossibilities, we follow a consolidated beyond-worst-case approach \citep{HaghtalabRS24}, which has already been successfully employed for bilateral trade \citep{AzarFF24}, and study smooth distributions. We detail how smoothness helps recovering suitable $\Mnet$, achieving the three desiderata sketched above; the argument for gain from trade is more streamlined than that for profit, so we focus on the latter.

    \paragraph{Smoothness Implies Fixed Discretizations}

        For profit maximization, we consider a fixed discretization of the mechanism space which relies on a suitable tiling of the $[0,1]^2$ square (see the family of simple mechanisms introduced in \Cref{subsec:restriction_profit}). We argue that the best mechanism in such $\Mnet$ is nearly optimal, under the smoothness assumption. Note, while this discretization is reminiscent of the ones used in \citet{AggarwalBDF24,DiGregorioDFS25}, we derive original bounds for the generic $p^{\textnormal{th}}$ moment of the approximation error (see \Cref{lem:lp_net}) which are also used to prove the novel uniform convergence bound for profit-maximization with smooth distributions (see Appendix~\ref{app:uniform_convergence}). 
        
    \paragraph{Smoothness Implies Fast Concentration}
    The second, and technically most involved challenge is to show that all mechanisms in $\Mnet$ have their profit well estimated, up to the $\alpha$ tolerated error. Ignoring lower-order terms, the cardinality of $\Mnet$ is exponential in $\nicefrac{1}{\alpha}$, so that combining Hoeffding inequality with a standard union bound would entail a (suboptimal) bound on the sample complexity bound of order $\nicefrac{1}{\alpha^2} \times \log |\Mnet| \approx \nicefrac{1}{\alpha^3}$. We improve on this bound via probabilistic chaining \citep{Talagrand14}: for each mechanism $M \in \Mnet$ we define a sequence of approximating mechanisms of increasing ``complexity'' and then perform a careful union bound along these chains (we refer to \Cref{fig:approximations_prof} for a visualization of approximating sequences). We mention that probabilistic chaining is also used in \citet{DiGregorioDFS25} to achieve tight sample complexity bounds for general distributions, however, our use of the mathematical tool is essentially different: there, the authors design a ``filtering procedure'' to ensure higher moments bounds for a distribution-and-data-dependent chaining decomposition, while here we exploit smoothness to make sure that the chaining analysis carries over on a fixed discretization, via uniform bounds on the errors' $p^{\textnormal{th}}$ moments.

    \paragraph{Computational Efficiency} The last challenge we need to overcome arises from the need to sample efficiently from the exponential distribution on $\Mnet$, whose cardinality is exponential in $\nicefrac{1}{\alpha}$. Indeed, in general exponential sampling leads to a running time linear in the support of the distribution. We exploit the geometric structure of the mechanisms to associate each simple mechanism with a path on a suitable graph supported in the $[0,1]^2$. We show how to associate weights to its (polynomially many) edges, so that we can simulate perfectly the exponential sampling on mechanisms/paths via random walks on the graphs, by sampling one edge after the other according to these weights. As already mentioned, this approach is inspired from an efficient sampling technique developed in the online learning literature \citep{TakimotoW03}.

    \subsection{Related Works}

        \paragraph{Learning in Bilateral Trade and Economics} A recent line of research adopts the lens of online learning to investigate the learnability of bilateral trade. They study a repeated setting, where the learner interacts with a sequence of bilateral trade rounds. There, buyer/seller values are either drawn from a distribution, or generated adversarially, with the objective of learning a mechanism from data that minimizes regret relative to the best mechanism in hindsight for either gain-from-trade \citep{Cesa-BianchiCCF24,Cesa-BianchiCCF24jmlr,AzarFF24, BernasconiCCF24,LunghiCM26} or profit \citep{DiGregorioDFS25}.  
        We also mention the existence of a rich literature on sample complexity bounds and regret rates for economically motivated problems \citep[e.g.,][]{ColeR14,MorgensternR15,Devanur0P16,CaiD17,GonczarowskiW21,CesaBianchiCK25,Cesa-BianchiGM15}, with some results focusing on smooth distributions as well \citep{DurvasulaHZ23,AggarwalBDF24}.
 
    \paragraph{Differential Privacy and Machine Learning.} 
    The first attempts at integrating differential private learning in statistical learning theory date back to the works of \citet{BlumDMN05,KasiviswanathanLNRS11}. 
    While there are instances in which private learning offered new algorithmic possibilities compared to normal learning \cite{CummingsLNRW16,BassilyNSSSU21,BassilyTT18}, private learning algorithms typically incur an overhead in the sample size and often require large running times compared to algorithms for non-private learning. This is why most research has focused on understanding the ``price of privacy''. Given its paradigmatic simplicity, much of the works on private learning revolves around binary classification in the standard offline PAC learning model \citep{KasiviswanathanLNRS11}, even though other settings such as convex optimization \cite{ChaudhuriMS11} have been investigated. 
    If the privacy guarantee only extends to the labels, the VC dimension characterizes private learnability in binary classification \citep{BeimelNS21}, as is also the case for non-private learning of classifiers. If the entire training set is required to be private---as in this paper---the sample size may increase considerably, with known bounds depending on the domain size \cite{BunNSV15}, the doubling dimension of the function space \cite{ChaudhuriH11}, or the representation dimension \cite{BeimelNS19}. \citet{AlonLMM19,AlonBLMM22} show that the Littlestone dimension, which characterizes online learnability, also characterizes private learnability; however their results are qualitative, so that the exact sample complexity is, to the best of our knowledge, not known. The best bounds are due to \citet{KaplanLMNS20}. We note that none of these results immediately imply much for our problem, as we are trying to learn the best mechanism, as opposed to binary classification. Actually, even the range space of generic monotone allocation regions has infinite VC dimension (and therefore also infinite Littlestone dimension, see \citet{DiGregorioDFS25}).

    \paragraph{Differential Privacy in Mechanism Design.} The interplay of differential privacy and mechanism design has been studied from various angles \citep[see, e.g.,]{PR13}. A fundamental connection exists between differentially private mechanisms and approximately truthful ones, as both require stability against changes in a single agent's private type. \citet{McSherryT07} formalizes this relation by showing that an $\eps$-differentially private mechanism is $2\eps$-approximately dominant strategy truthful. Such a characterization enables, for example,  the design of digital auctions that extract higher revenue than the best strictly truthful mechanisms. Furthermore, \citet{NST12} demonstrates that differential privacy can facilitate exact dominant strategy truthfulness in settings without monetary transfers. A second line of research treats privacy preservation as a primary design objective. For instance, \citet{HK12} and \citet{ChenCKMV16} provide private generalizations of the VCG mechanism, offering a smooth trade-off between the privacy budget and the approximation ratio while guaranteeing exact dominant strategy truthfulness. Similarly, the Sensitive Surveyor's Problem \citep{Rot12} explores conducting studies on private data while explicitly compensating individuals for the cost of privacy loss. Our approach differs fundamentally from prior work by focusing on \emph{learning} nearly-optimal mechanisms, in a differentially private way.

\section{Preliminaries}
\label{sec:preliminaries}
    
    In the bilateral trade problem, a seller holds a single good that a buyer wants to purchase. The seller's and buyer's valuations for the good are $\vs \in [0,1]$ and $\vb \in [0,1]$, respectively, and are private information. Seller and buyer submit bids $\bs, \bb \in [0,1]$ (not necessarily truthful) to a broker running mechanism $M$, characterized by an allocation region $A \subseteq[0,1]^2$, and pricing rules $p ,q:[0,1]^2 \to [0,1]$. 
    The trade happens if and only if the bids $(\bs,\bb)$ belong to the allocation region $A$, and payments are made according to $p$ and $q$. 
    That is, the seller receives $p(\bs,\bb)$, 
    while the buyer pays $q(\bs,\bb)$.
    Without loss of generality, we require that $p(\bs,\bb) = q(\bs,\bb) = 0$ whenever $(\bs,\bb) \not\in A$, i.e., there is no trade.

    The seller and the buyer are rational agents with quasi-linear utility. Namely, when they submit bids $(\bs, \bb) \in [0,1]^2$ but have valuations  $\vb$ and $\vs$, the utilities are, respectively: 
    \[
        \util{\bs}{\bb}{s} =  \ind{(\bs,\bb)\not \in A} \cdot \vs+ p(\bs, \bb), \quad \util{\bs}{\bb}{b} = \ind{(\bs,\bb)\in A} \cdot \vb - q(\bs, \bb).
    \]
     We focus on dominant-strategy incentive compatible (DSIC) and individually rational (IR) mechanisms: each player maximizes his utility by a truthful bid regardless of the other player's bid, and the utility from participating in the mechanism is at least as high as that from not participating in the mechanism (i.e., it is at least $\vs$ for the seller and at least $0$ for the buyer).
    In formulas,
    \begin{align*}
        \text{DSIC:}\quad \util{\vs}{\bb}{s} &\geq \util{\bs}{\bb}{s} &&\forall \ \vs \in [0,1], (\bs, \bb) \in [0, 1]^2 \\
        \util{\bs}{\vb}{b} & \geq \util{\bs}{\bb}{b} &&\forall \ \vb \in [0,1], (\bs, \bb) \in [0, 1]^2 \nonumber \\
        \text{IR}:\quad\;\;\; \util{\vs}{\bb}{s} &\geq \vs, \,\, \util{\bs}{\vb}{b} \geq 0 &&\forall \ (\vs, \vb) \in [0, 1]^2, (\bs, \bb) \in [0, 1]^2 \nonumber
    \end{align*}

    We denote the class of all DSIC and IR mechanisms by $\M$. Since we focus on such class, in the rest of the paper, we assume that the players' bids are equal to their respective valuations. We consider two natural objectives: gain-from-trade and profit. The gain-from-trade measures the increase in social welfare resulting from the mechanism, and is a measure of economic efficiency:
    \[
        \gft(M,(\vs,\vb)) = (\vb - \vs) \cdot \ind{(\vs,\vb)\in A} 
    \]
    The profit earned by the broker running mechanism $M$ 
    is the difference between the payment the mechanism collects from the buyer and the payment that it makes to the seller:
    \[
        \prof(M,(\vs,\vb)) = q(\vs,\vb) - p(\vs,\vb). 
    \]
    A mechanism enforces \emph{budget balance} if the profit is always non-negative. For brevity, when it is clear from the context, we may adopt the vectorial notation and replace $(\vs,\vb)$ with $v.$

    \subsection{Characterization of DSIC Mechanisms}

The allocation regions and the pricing functions of DSIC and IR mechanisms have a distinctive structure: the allocation region respects a certain monotonicity property, while the allocation region uniquely induces the payments.

    \begin{definition}[Monotone Allocation Region]
    \label{def:monotone_allocation}
        An allocation region $A\subseteq[0,1]^2$ is monotone  if for any $x=(x_1,x_2) \in A$ and $y=(y_1,y_2) \in [0,1]^2$ with $x_1 \ge y_1$ and $x_2 \le y_2,$ it holds that $y \in A$. A mechanism is monotone if its allocation region is monotone.
    \end{definition}

    \begin{definition}[Myerson Payments]
    \label{def:payments}
        Let $A$ be a monotone allocation region, then the associated Myerson payments are defined as follows:
        \begin{align*}
            p(\bs,\bb) =& \ind{(\bs,\bb) \in A} \cdot \max\{x \in [0,1] \mid (x,\bb) \in A\} &&\forall\; (\bs,\bb) \in [0,1]^2\\
            q(\bs,\bb) = &\ind{(\bs,\bb) \in A}\cdot \min\{y \in [0,1] \mid (\bs,y) \in A\} &&\forall\; (\bs,\bb) \in [0,1]^2 
        \end{align*}
    \end{definition}

    In words, for an allocation region to be monotone it should be ``closed'' in a north-west direction. That is, if a point $(\bs,\bb)$ is in $A$ then any point $(\bs',\bb)$ with $\bs' \leq \bs$ and any point $(\bs,\bb')$ with $\bb' \geq \bb$ should be in $A$ as well. The payments of a point $(\bs,\bb) \in A$, in turn, correspond to the ``east'' projection minus the ``south'' one onto the allocation boundary. See \Cref{fig:payments} for an illustration. Note that for the above definition to be well-posed, we require (without loss of generality) that all allocation regions are topologically closed.

    \begin{figure}[t!]
      \centering
      \begin{subfigure}{0.30\linewidth}
        \centering
        \includegraphics[width=\linewidth]{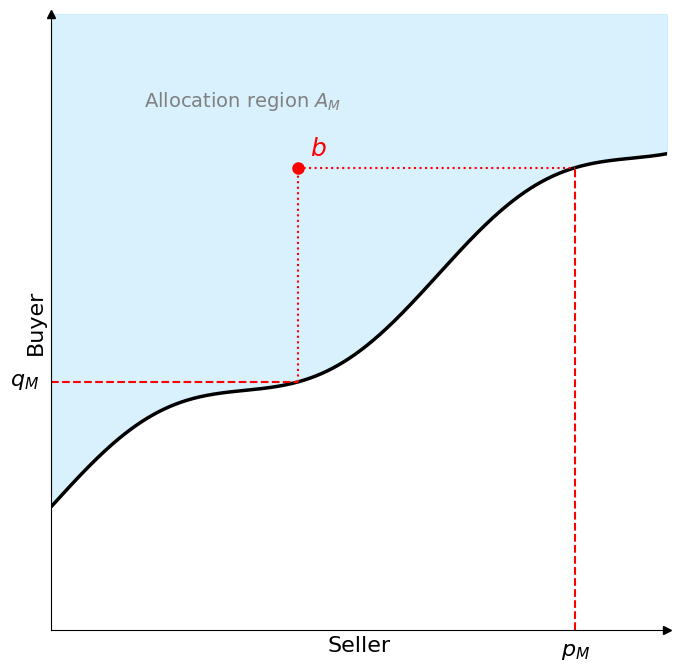}
        \label{fig:generic_mechanism}
      \end{subfigure}
      \hspace*{26pt}
      \begin{subfigure}{0.30\linewidth}
        \centering
        \includegraphics[width=\linewidth]{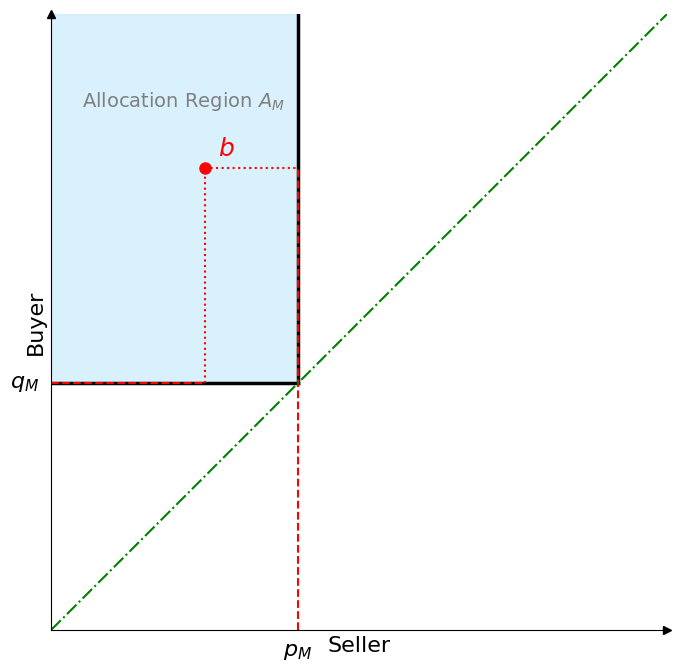}
        \label{fig:fixed_price}
      \end{subfigure}
    \caption{Payments and allocation regions for a generic monotone mechanism and a fixed price mechanism.}
    \label{fig:payments}
    \end{figure}
    
    \begin{restatable}[Mechanisms Characterization]{proposition}{characterization}
    \label{thm:mechanisms}
        A mechanism $M$ for bilateral trade is dominant-strategy incentive compatible and individually rational if and only if its allocation region is monotone and the payments are Myerson payments.
    \end{restatable}

    The proof of \Cref{thm:mechanisms} follows by standard arguments \citep{Myerson81,MyersonS83}, and is provided explicitly in \citet{DiGregorioDFS25}. A crucial consequence of \Cref{thm:mechanisms} is that it reduces the learning problem to finding the best monotone allocation region $A$. When the goal is maximizing gain from trade while enforcing budget balance, the mechanism space collapses to the simple family of fixed-price mechanisms \citep{Colini-Baldeschi16,Hagerty87}. A fixed price mechanism $M_p \in \M$ is characterized by a price $p \in [0,1]$, which is posted to both agents. The trade happens if they both accept (i.e., $\vs \le p \le \vb$), with the buyer paying exactly $p$ to the seller. In other words, $M_p$ allocates in the rectangle $[0,p] \times [p,1],$ see also \Cref{fig:payments}. Clearly, fixed-price mechanisms enforce budget balance. We have the following Proposition, whose simple proof is deferred to Appendix~\ref{app:preliminaries}.
    \begin{restatable}{proposition}{fixed}{\textnormal{(Fixed Price Mechanisms)}}
    \label{prop:fixed_price}
        Consider any mechanism $M \in \M$ that enforces budget balance, then there exists a fixed-price mechanism $M_p \in \M$ such that the following holds for any distribution $\cD$
        \[
            \Esub{v \sim \cD}{\gft(M,v)} \le \Esub{v \sim \cD}{\gft(M_p,v)}
        \]
    \end{restatable}

    \subsection{The Learning Task}
        The goal of this paper is to design learning algorithms for profit and gain-from-trade maximization in bilateral trade that are differentially private. Given sample access to a distribution $\cD$ over $[0,1]^2$, we want to find mechanisms that maximize the expected profit or gain-from-trade, when valuations are drawn according to $\cD$. Note that the agents' valuations may be correlated. In the following and henceforth, we explicitly mention the variables and/or distributions with respect to which probability statements are made only when there is any ambiguity. If the context is clear, we omit this information. In general, the sources of randomness we consider are limited to three: the randomness inherent in the sample $S$, the randomness of the fresh sample $v$ drawn from the same distribution, and the internal randomness of the differentially private algorithm. 
        
        \begin{definition}
            Let $\acc > 0$ and $\fail \in [0,1]$. 
            Let $\cA$ be a possibly randomized learning algorithm that takes as input i.i.d. samples $S$ from $\cD$ and outputs a mechanism in $\M$. We say that $\cA$ is $(\acc,\fail)$-optimal with respect to metric $f$ if
            \[
                \Psub{\cA,S}{\sup_{M \in \M}\Esub{v\sim \cD}{f(M,v) - f(\cA(S),v)} \le  \acc} \ge 1-\fail.
            \]
        \end{definition}

        This definition is a standard ``PAC-learning'' style one: with probability $1-\fail$ with respect to the sampling phase and its internal randomization, the algorithm $\cA$ finds a mechanism $\cA(S)$ that is at most an additive $\acc$ apart from the best mechanism, \emph{on a fresh sample}.
        An important family that we study in this paper is that of the $\sigma$-smooth distributions \citep[e.g.,][]{HaghtalabRS24}, characterized by density functions that are not ``too peaked''. 
        \begin{definition}[Smooth Distributions]
        \label{def:smoothness}
            A distribution $\cD$ over $[0,1]^2$ is said to be $\sigma$-smooth, for some $\sigma \in (0,1]$, if the following inequality holds for any Borel set $A$ in $[0,1]^2$:
            \[
                \Psub{v \sim \cD}{v \in A} \le \frac{\cL(A)}{\sigma},
            \]
            where $\cL$ denotes the Lebesgue measure. A random variable is $\sigma$-smooth if its distribution is $\sigma$-smooth.
        \end{definition}

    \subsection{Differential Privacy}

        We now recall some notions from the differential privacy literature \citep{DworkR14}. Suppose $\mathcal{S}$ is a given family of possible data sets\footnote{We allow the dataset to contain the same point multiple times, so that $\cS$ may contain multisets.} over some domain, we say that two data sets $S,S'\in \mathcal{S}$ are neighboring if they differ in exactly one data point.         
        \begin{definition}
        Let $\varepsilon>0$. A randomized algorithm $\mathcal{A}: \mathcal{S}\rightarrow \mathcal{Y}$ is $\varepsilon$-differentially private, if for all neighboring data sets $S,S'\in \mathcal{S}$ and all outputs $M\in \mathcal{Y}$
        \[\mathbb{P}\left[\cA(S)=M\right] \leq e^{\varepsilon}\mathbb{P}\left[\cA(S')=M\right],\]
        where the probabilities are over the randomness of $\mathcal{A}$.
        \end{definition}
        
            For simplicity, we state the definition of differential privacy for \emph{finite} target spaces $\cY$, used for all our positive results. However, for our impossibility results we also need the one for general target spaces, which is fairly immediate\footnote{Given a $\sigma$-algebra $\Sigma$ on $\cY$, a measurable randomized algorithm is $\eps$-differentially private if \(\P{\cA(S) \in \tilde \M} \le e^{\eps} \P{\cA(S') \in \tilde \M}\) for all $\tilde \M \in \Sigma.$}. In this paper, our goal is to learn (good) mechanisms for bilateral trade; hence, the data sets are i.i.d. valuations of the type $v = (\vs,\vb)$ drawn from a distribution $\cD$ supported over $[0,1]^2$, and the output space is the family of mechanisms $\M$. A crucial ingredient in our approach is the \emph{exponential mechanism} \citep{McSherryT07}, which relies on the notion of sensitivity.

    \begin{definition}[Sensitivity]
        Let $f:\mathcal{S}\times \mathcal{Y} \rightarrow \mathbb{R}_{\geq 0}$. The sensitivity of $f$ with respect to two neighboring data sets $S$ and $S' \in \cS$ is 
        \[
            \Delta^{S,S'}_f = \max_{M\in \mathcal{Y}}\left\vert f(S,M) - f(S',M) \right\vert.    
        \]
        The sensitivity $\Delta_f$ of $f$ is given by the largest sensitivity for any neighboring sets: $\sup \Delta^{S,S'}_f$ for $S$ and $S'$ neighboring data sets in $\cS$.
    \end{definition}

    In this paper, we consider finite target space $\cY$. Given some score function $f$ with sensitivity $\Delta_f$, the exponential mechanism consists in sampling from $\cY$ with probability proportionate to $\exp\left(\varepsilon\cdot \frac{f(S, M)}{2 \Delta_f}\right)$. That is, the sampling algorithm $\expalg_f$ selects $M$ with probability
    \[
        \P{\expalg_f(S) = M} = \frac{\exp\left({\eps\cdot\frac{f(S,M)}{2 \Delta_f}}\right)}{\sum_{M'\in \mathcal{Y}}\exp\left({\eps\cdot\frac{f(S,M')}{2 \Delta_f}}\right)}.
    \]
    The exponential mechanism is $\eps$-differentially private \citep{McSherryT07}.

    \subsection{A Meta Theorem}
    \label{subsec:meta}
        We present a Meta-Theorem that spells out requirements to achieve, at the same time, approximate optimality and differential privacy. Our approach consists in finding a suitable subfamily $\Mnet$ of mechanisms, and then running the exponential mechanism on it, using as score function the empirical performance on samples. More precisely, we require a ``uniform convergence'' bound on $\Mnet$, and that such a family is representative enough of the whole mechanism space, i.e., that is not too far from the optimal mechanism.
        Since the same argument works for both gain-from-trade and profit, we use $f$ as a proxy for both objectives, and $ f(\cdot, S)$ as its empirical estimate on dataset $S$:
        \[
            f(M, S) = \frac{1}{n} \sum_{v \in S} f(M,v).
        \]

        Given a finite sub-family $\Mnet$ of mechanisms, we consider the exponential mechanism supported on $\Mnet$ that uses $f(\cdot, S)$ as its score function. Since both profit and gain from trade have values in $[-1,1]$, then the sensitivity of $f$ is upper bounded by $\nicefrac 2n.$ In the following meta-theorem we formalize the requirements on $\Mnet$ to achieve $(\acc,\fail)$-optimality and $\eps$-differential privacy, along with the corresponding sample complexity required. 

        \begin{restatable}{theorem}{meta}{\emph{(Meta-Theorem)}}
        \label{thm:meta}
            Fix any $\acc>0$, $\fail \in [0,1]$, and $\eps > 0$, and consider the objective $f$. Let $S$ be a random i.i.d. sample of size $n$, with $v$ independently drawn from the same distribution.
            Assume that there exists a finite family $\Mnet \subseteq \M$ that respects the following properties:
            \begin{itemize}
                \item[(i)] It holds that
                \[
                    \sup_{M \in \M} \Esub{v}{f(M,v)} - \max_{M \in \Mnet} \Esub{v}{f(M,v)} \le\frac{\acc}4    
                \]
                \item[(ii)] There exists $n_0 = n_0(\acc,\fail)$ such that when $|S| \ge n_0$, the following inequality holds with probability at least $1-\nicefrac{\fail}2$ over the sampling of $S$:
                \[
                    \max_{M \in \Mnet}|\Esub{v}{f(M,v)} - f(M, S)| \le \frac{\acc}4.
                \] 
            \end{itemize}
            If $|S| \ge  n_0+ \frac{16\log(\nicefrac{2\lvert \Mnet\rvert}{\fail})}{\acc\eps}$, then there exists a randomized learning algorithm that is $(\acc,\fail)$-optimal and $\eps$-differentially private.
        \end{restatable}

        The proof of the theorem is deferred to Appendix~\ref{app:meta}. We conclude with some comments on the two properties required by the meta-theorem. Property (i) states that the subfamily $\Mnet$ contains a mechanism whose performance is $\acc$-close to the optimal value that $f$ may achieve, i.e., that $\Mnet$ has a small \emph{discretization error} with respect to $f$. On the other hand, property (ii) requires that we have enough samples to obtain \emph{uniform convergence} in the subfamily; stated differently, the objective $f$ is well estimated over all the mechanisms in $\Mnet$.

\section{Profit Maximization}

    In this Section, we investigate private and nearly optimal mechanisms for profit maximization. First, in \Cref{subsec:lower_profit}, we construct a hard family of distributions for which optimality and differential privacy are incompatible. Then, in the rest of the section, we construct the desired mechanisms assuming smoothness of the underlying distribution. In \Cref{subsec:wrapping} we achieve this result by following the pipeline from \Cref{thm:meta}: in \Cref{subsec:restriction_profit}, we present a finite family of mechanisms that well approximates the optimal profit, while in \Cref{subsec:unif_convergence}, we show its uniform convergence properties. 
    
    \subsection{Private Learning is Impossible for General Distributions}
    \label{subsec:lower_profit}

    We prove that private learning of profit-maximizing mechanisms for bilateral trade is impossible.

    \begin{theorem}
    \label{thm:impossibility_profit}
        For any learning algorithm $\cA$ for profit maximization in bilateral trade, precision parameter $\alpha < \nicefrac{1}{64}$, privacy parameter $\varepsilon>0$ and number of samples $N$, there exists a hard distribution $\cD$ for which the learning algorithm $\cA$ on $N$ i.i.d. samples cannot be $\eps$-differentially private and satisfy: 
        \begin{equation}
        \label{eq:optimality_expectation}
                \sup_{M \in \M} \E{\prof(M,v)} \le \E{\prof(\A(S),v)} + \alpha,
        \end{equation}
            where the expectation is with respect to $S$ and $v$ drawn from $\cD$, and the randomization of $\A$. 
    \end{theorem}
    \begin{proof}
        Let $\alpha' = 2\alpha$ and consider a family of hard input distributions as follows. Let $\delta$ be a small parameter we set later (for simplicity, we assume $\nicefrac{1}{\delta}$ is integer) and define $\xi_i = \delta^{i}$, for $i\in\{1,\ldots,\delta^{-1}\}$. For any such $i$, we introduce $\mathcal{D}_i$, the uniform distribution on the line $\ell_i$ connecting the points $(0,\nicefrac 12)$ and $(\xi_i,1)$. By construction, the optimal mechanism $M_i$ for profit on $\mathcal{D}_i$ allocates in the triangle $(0,\nicefrac 12)$--$(\xi_{i},1)$--$(0,1)$, yielding an (expected) profit of $ \int_{0}^{\xi_i} \left(\nicefrac{1}{2}+(\nicefrac{1}{(2\xi_i)} -1)x\right)\frac{1}{\xi_i} dx = \nicefrac 34-\nicefrac{\xi_i}{2}$. Crucially, mechanisms which performs well on a given $\cD_i$ are suboptimal for the other distributions of the family, as formalized in the following claim, whose proof is deferred to Appendix \ref{app:missing_lower_profit}.
\begin{restatable}{claim}{lbprofits}{}
\label{claim:lbprofits} Let $\alpha',\delta<\nicefrac{1}{32}$. Consider any $i,j\in\{1,\ldots,\delta^{-1}\}$, with $i\neq j$. If $\Esub{v \sim \cD_i}{\prof(M,v)} \ge \nicefrac 34 -\nicefrac{\xi_i}{2}-\alpha'$, then $\Esub{v \sim \cD_j}{\prof(M,v)}\le  \nicefrac{2}{3} < \nicefrac 34 -\nicefrac{\xi_j}{2}-2\alpha'$.
\end{restatable}
        Stated differently it $\mathcal{M}_i$ denotes the set of mechanisms that have profits competitive with $M_i$ on $\mathcal{D}_i$ (i.e., whose expected profit under $\cD_i$ is at least $\nicefrac 34 - \nicefrac{\xi_i}2 -\alpha'$), then any $M \in \mathcal{M}_i$ cannot be competitive with $M_j$ on $\mathcal{D}_j$. 

Now, assume towards contradiction that there exists a $\eps$-differentially private learning algorithm $\cA$ satisfying \Cref{eq:optimality_expectation} whenever the valuations are sampled by any distribution in the hard family. Therefore, if the underlying distribution is $\mathcal{D}_i$ (for an arbitrary $i \in \{1, \ldots, \delta^{-1}\}$), the algorithm outputs a mechanism $M\in \mathcal{M}_i$ with expected profit at least $\nicefrac{3}{4}- \nicefrac{\xi_i}{2} - \alpha'$ with probability at least $\nicefrac{1}{2}$. To see why this is the case, notice that by assumption we have: 
$$
\Esub{\cA, S\sim \cD^N_i}{\Esub{v\sim \cD_i}{\prof(M_i, v)-\prof(\cA(S), v)}}\leq \alpha, 
$$
which implies, by Markov's inequality, since $M_i$ is the optimal mechanism for $\cD_i$,
\begin{equation*}
\Psub{\cA, S\sim \cD_i^N}{\cA(S)\notin \M_i} = \Psub{\cA, S\sim \cD_i^N}{\Esub{v\sim \cD_i}{\prof(M_i, v)-\prof(\cA(S), v)} > \alpha'} 
\leq \frac{\alpha}{\alpha'} = \frac{1}{2}
\end{equation*}

Now denote by $S_i^N$ a sample that only draws $N$ points from distribution $\mathcal{D}_i$ and by $S_{i,j}^{k}$ a sample where $k$ of the points drawn from $\mathcal{D}_i$ are replaced by points drawn from $\mathcal{D}_j$. 
Thus, by definition of differential privacy we have the following chain of inequalities:
            \begin{align}
               \frac{1}{4} <\, \P{\A(S_i) \in \mathcal{M}_i} &\le e^{\eps} \P{\A(S^1_{i,j}) \in \M_i}\tag{$S_i$ and $S^1_{i,j}$ differ in one point}\\
                &\le e^{\eps}\left[e^{\eps} \P{\A(S^2_{i,j}) \in \mathcal{M}_i}\right]  \tag{$S^1_{i,j}$ and $S^2_{i,j}$ differ in one point}\\
                &\le e^{N\eps}\P{\A(S_{j})\in \mathcal{M}_i}\tag{iterating $N$ times} 
            \end{align}
        Thus, we have $\nicefrac{1}{4}\cdot e^{-N\varepsilon} < \P{\A(S_{j})\in \mathcal{M}_i}$.
        Summing up over all candidate distributions $\mathcal{D}_i$ with $i\neq j$, we have 
        \[
            \P{\cup_{i\neq j}\{\cA(S_j)\in \M_i\}} = \sum_{i\neq j} \P{\A(S_{j}) \in \mathcal{M}_i} > e^{-N\varepsilon}\frac{1}{4} \left(\frac{1}{\delta}-1\right)\geq e^{-N\varepsilon}\cdot \frac{1}{8\delta}, 
        \]
        where the first equality follows by the pairwise disjointedness of $\{\M_i\}_{i=1}^{\delta^{-1}}$, which is implied by \Cref{claim:lbprofits}. 
        We can now choose $\delta < \min \{\nicefrac{1}{32}, \nicefrac{e^{-N\varepsilon}}{4}\}$, so that \Cref{claim:lbprofits} holds and the right-hand-side term of the previous display is greater than $\nicefrac{1}{2}$. Thus with probability $p$ at least $\nicefrac{1}{2}$, the algorithm should release a mechanism from $\cup_{i\neq j}\mathcal{M}_i$ on distribution $\mathcal{D}_j$. By the law of total probability and \Cref{claim:lbprofits}, the expected profit of $\cA$ on $\mathcal{D}_j$ is then no more than: 
        $$
        p\cdot \nicefrac{2}{3}+(1-p)\cdot \opt_j = \opt_j - p(\opt_j-\nicefrac{2}{3}) < \nicefrac{1}{2}\cdot \opt_j + \nicefrac{1}{3},
        $$
        where $\opt_j = \Esub{v\sim\cD_j}{\prof(M_j, v)}$.
        Taking the difference between $\opt_j$ and the above we get 
        \[
            \nicefrac{1}{2}\cdot \opt_j - \nicefrac{2}{6} = \nicefrac{1}{2}\cdot(\nicefrac{3}{4}-\nicefrac{\xi_j}{2})-\nicefrac{2}{6}> \nicefrac{1}{2}\cdot(\nicefrac{3}{4}-\nicefrac{1}{64})-\nicefrac{2}{6}>\nicefrac{1}{64} > \alpha.
        \]
        Therefore, \Cref{eq:optimality_expectation} cannot hold for $\cA$ when $\cD_j$ is the underlying valuation distribution, yielding a contradiction.
        
    \end{proof}

    \subsection[The eta-simple Mechanisms]{The $\eta$-simple Mechanisms}
    \label{subsec:restriction_profit}

        We introduce a discrete family of mechanisms, parameterized by a precision $\eta$ which we assume, for convenience, to be a power of $2$, i.e., $\eta = 2^{-H}$, for some positive integer $H$. This assumption is without loss of generality, as rounding the precision to the closest power of $2$ only affects the sample complexity by a constant factor. Consider the uniform grid $V_{\eta}$ of step-size $\eta$ of $[0,1]^2$, i.e., the points of the form $(k \eta, j\eta)$ for $k$ and $j$ in $0, \dots, \nicefrac{1}{\eta}$. The grid naturally divides the $[0,1]^2$ square into squared ``tiles'' of the form $[k \eta, (k+1)\eta] \times [j \eta, (j+1)\eta]$, we denote with $\cT_{\eta}$ the set of such tiles. We define the family of $\eta$-simple mechanisms $\M_{\eta}\subseteq \M$ as all the mechanisms whose (monotone) allocation region is given by the union of tiles in $\cT_{\eta}$ (see also \Cref{fig:approximations_prof}). By definition, the allocation regions of $\eta$-simple mechanisms are given by the union of segments joining points in the underlying uniform grid, this allows to bound the cardinality of $\M_{\eta}$ via a simple combinatorial argument which is deferred to Appendix~\ref{app:missing_net}.  

        \begin{restatable}{lemma}{cardbound}   \label{prop:card_bound}
            The following bound on the cardinality of $\eta$-simple mechanisms holds: $\lvert \M_\eta\rvert \leq {(2e)^{\nicefrac{1}{\eta}}}$. 
        \end{restatable}

        The power of $\eta$-simple mechanisms resides in the fact that they nicely approximate the profit of any mechanism, as long as the underlying distribution is smooth. The approximating map is fairly natural: given a generic mechanism $M$, just consider the union of all the tiles in $\cT_{\eta}$ which are fully contained in the allocation region of $M$, some sort of $\eta$-simple ``inner'' approximation. The definition below is illustrated in \Cref{fig:approximations_prof}.
        
        \begin{definition}[$\eta$-approximations]
        \label{def:approximations_prof}
        For any precision parameter $\eta$ and any $M \in \M$ with allocation region $A$, we define $\app{\eta}{M} \in \M_{\eta}$ as the mechanism with allocation region given by the union\footnote{If the union is empty, we conventionally let the allocation region be the single point $(0, 1)$.} of all $T \in \cT_{\eta}$ such that $T \subseteq A$.
        \end{definition}
        
        We can then prove the approximation result needed in the first point of the meta-theorem, and complement it with a stronger bound on the general $p^{\textnormal{th}}$ moment of the approximation errors. 
        
    \begin{lemma}
    \label{lemma:restriction}
    Let $v$ be a $\sigma$-smooth random variable on $[0, 1]^2$. For any precision $\eta>0$, the following bound holds:
    \begin{equation*}
        \max_{M\in \mathcal M}\E{\prof(M, v)} - \max_{M\in \mathcal M_\eta} \E{\prof(M, v)} \le \frac{2\eta}{\sigma}
    \end{equation*}
    \end{lemma}
    \begin{proof}
        Let $M^\star$ be the profit-maximizing mechanism for the given distribution. $M^\star$ exists, since the maximum is attained in $\M$, as we prove in Appendix \ref{sec:max_attained}. Let $\phi^{\star}$ the $\eta$-simple mechanism that approximates it, i.e. $\phi^{\star}=\app{\eta}{M^{\star}}$. Denote with $A_{M^\star}$, respectively $A_{\phi^\star}$ their allocation regions. By definition of $\phi^\star$, the allocation region of the former mechanism contains the latter, which has an important impact on the prices proposed by the two mechanisms. Indeed, consider any point $v \in A_{\phi^\star}$, by the definition of Myerson payments and the fact that $A_{\phi^\star}\subseteq A_{M^\star}$, it follows that $\prof(M^\star,v) \le \prof(\phi^\star,v)$. We then have the following chain of inequalities:
        \begin{align*}
            \E{\prof(M^\star, v)} &=\E{\prof(M^\star, v)\ind{v\in A_{\phi_\eta(M^\star)}}} + \E{\prof(M^\star, v)\ind{v\in A_{M^\star}\setminus A_{\phi_\eta(M^\star)}}}\\
            &\leq\E{\prof(\phi_\eta(M^\star), v)\ind{v\in A_{\phi_\eta(M^\star)}}} + \P{v\in A_{M^\star}\setminus A_{\phi_\eta(M^\star)}}\\
            &\le \max_{M\in \M_\eta}\E{\prof(M, v)} + \P{v\in A_{M^\star}\setminus A_{\phi_\eta(M^\star)}},
        \end{align*}
        where the inequality follows from the fact that the profit is always bounded above by $1$. In \Cref{cl:cardinality} (Appendix~\ref{app:missing_net}), we prove that $\P{v\in A_M\setminus A_{\phi_\eta(M)}} \le \nicefrac{2\eta}{\sigma}$; thus concluding the proof. 
        In a nutshell, \Cref{cl:cardinality} holds as there are at most $\nicefrac{2}{\eta}$ tiles in $\cT_{\eta}$ that have non-trivial intersection with $A_M$ (i.e., they intersect with $A_M$ but are not fully contained in it). Moreover, each such cell has probability at most $\nicefrac{\eta^2}\sigma$ by smoothness, and their union is a superset of $A_M\setminus A_{\phi_\eta(M)}$ (See also \Cref{fig:approximation_difference}).
    \end{proof}
        \begin{figure}[t!]
          \centering
          \begin{subfigure}{0.22\linewidth}
            \centering
            \includegraphics[width=\linewidth]{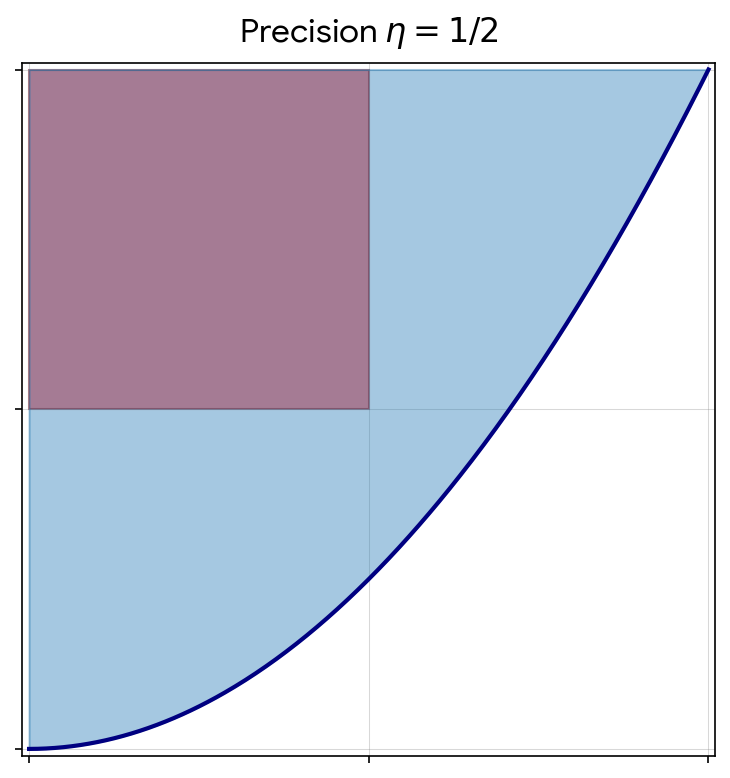}
          \end{subfigure}%
          \hspace*{6pt}
          \begin{subfigure}{0.22\linewidth}
            \centering
            \includegraphics[width=\linewidth]{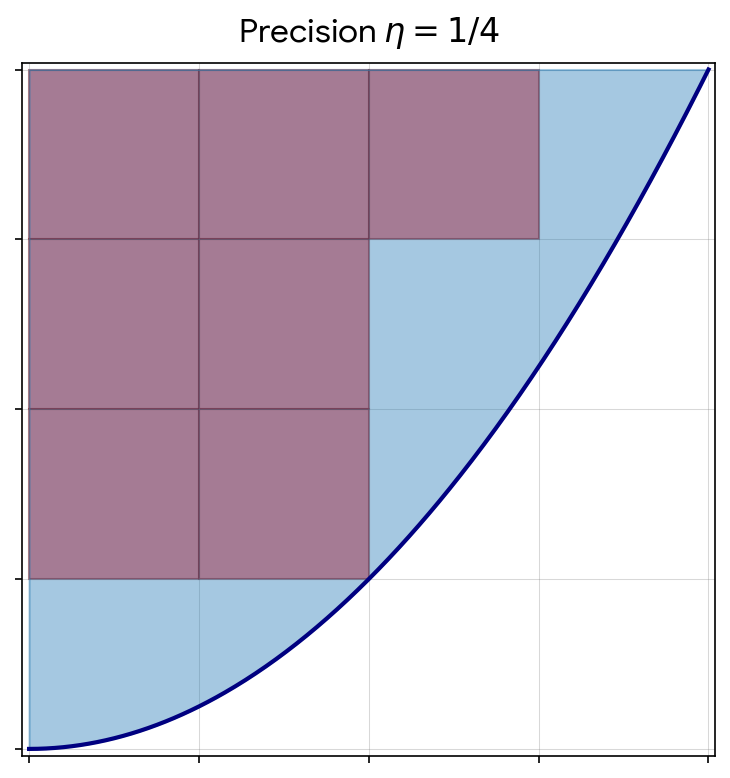}
          \end{subfigure}
          \hspace*{6pt}
          \begin{subfigure}{0.22\linewidth}
            \centering
            \includegraphics[width=\linewidth]{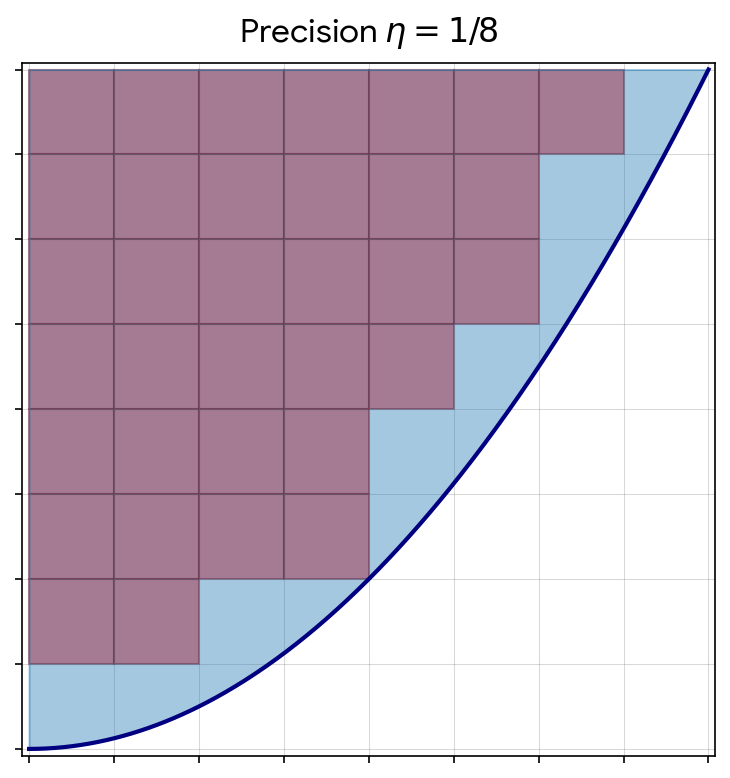}
          \end{subfigure}
          \hspace*{6pt}
          \begin{subfigure}{0.22\linewidth}
            \centering
            \includegraphics[width=\linewidth]{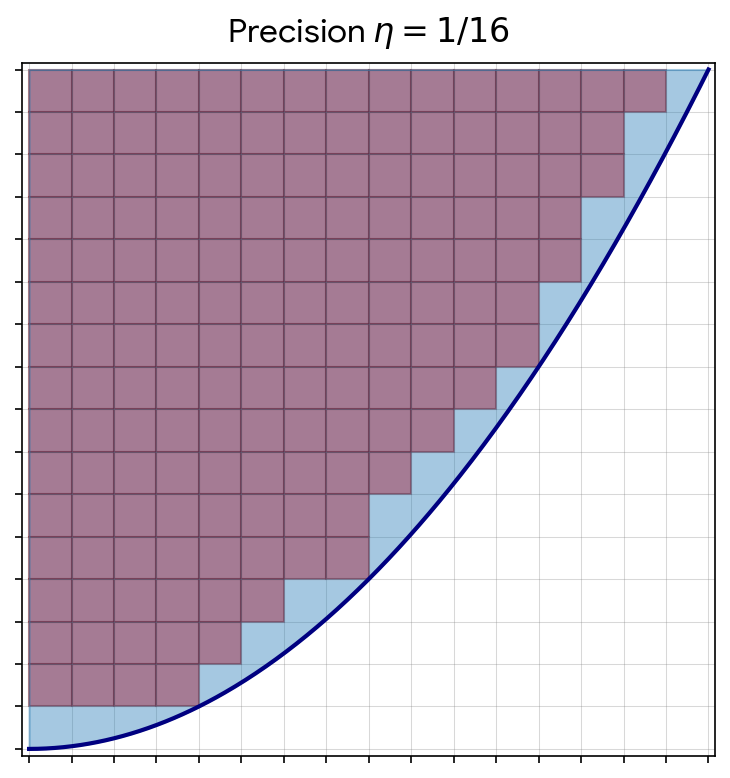}
          \end{subfigure}
        \caption{A progressive approximation of a mechanism with $\phi_\eta$.}
        \label{fig:approximations_prof}
        \end{figure}
    \begin{restatable}{lemma}{lpnet}
    \label{lem:lp_net}
    Let $v$ be a $\sigma$-smooth random variable on $[0, 1]^2$. For every $\eta > 0$ and $p \ge 2$: 
    \begin{equation*}
        \E{\lvert\prof(M, v)-\prof(\app{\eta}{M}, v)\rvert^p} \leq 2(2^{p+2}+1) \cdot \frac{\eta}{\sigma}
    \end{equation*}
    \end{restatable}
    \begin{proof}
        Consider a generic mechanism $M$ with allocation region $A_M$ and payments $p_M, q_M$, and its $\eta$-simple approximation $\app{\eta}{M}$, with allocation region $A_{\phi}$ and payments $p_{\phi}, q_{\phi}$. We can decompose the $p^{\text{th}}$ norm to the $p$-th power depending on whether the realized valuation falls within $A_{\phi}$ or in $A_M \setminus A_{\phi}$. In the latter case, the analysis is aligned with that of \Cref{lemma:restriction}, noting that the profit induced by $\app{\eta}{M}$ is zero outside its allocation region:
        \begin{equation}
        \label{eq:v_notin_Aphi}
            \E{\lvert\prof(M, v)-\prof(\app{\eta}{M}, v)\rvert^p\ind{v \in A_M \setminus A_{\phi}}} \le \P{v \in A_M \setminus A_{\phi}}\le \frac{2\eta}{\sigma},
\end{equation}
    where the last inequality follows by \Cref{cl:cardinality}.
         The argument for $A_{\phi}$ is more subtle, as the payments induced by the simple mechanism may be considerably different from those computed by $M$, as illustrated in \Cref{fig:error_buyer,fig:seller_error}. Since the profit is composed by the price paid by the buyer, corresponding to the vertical projection, and that paid to the seller, horizontal projection, we partition $[0,1]^2$ into columns $C$ and rows $R$. A column is simply a rectangle of the type $[k\eta, (k+1)\eta] \times [0,1]$, while a row has the form $[0,1] \times [j\eta, (j+1)\eta]$. We can associate to each column (respectively row), the maximum gap between the price paid by the buyer (respectively to the seller) in the two mechanisms:
        \(
            M_c = \sup_{v \in c\cap A_\phi} \left\{q_{\phi}(v) - q_M(v)\right\}\) and \(M_r = \sup_{v \in r\cap A_\phi} \left\{p_M(v) - p_{\phi}(v)\right\}.
        \)
        We can then partition the columns and rows according to their gap: $R_i = \{r \in R| 2^{-i-1} < M_r \le 2^{-i}\}$ and $C_i = \{c \in C| 2^{-i-1} < M_c \le  2^{-i}\}$, to get the following chain of inequalities:
        \begin{align*}
            &\E{\lvert\prof(M, v)-\prof(\app{\eta}{M}, v)\rvert^p\ind{v \in A_{\phi}}} \\
            &\le 2^{p-1}\E{\left(|q_{\phi}(v) - q_M(v)|^p + |p_M(v) - p_{\phi}(v)|^p\right)\ind{v \in A_{\phi}}} \tag{As $f(x)= |x|^p$ is convex for $p \ge 1$}\\
            &\leq 2^{p-1}\sum_{i=0}^{\infty}\E{\sum_{c \in C_i}|q_{\phi}(v) - q_M(v)|^p\ind{v \in A_{\phi} \cap c} + \sum_{r \in R_i}|p_M(v) - p_{\phi}(v)|^p\ind{v \in A_{\phi} \cap r}}\\
            &\leq 2^{p-1}\sum_{i=0}^{\infty}\E{\sum_{c \in C_i}M_c^p\ind{v \in A_{\phi} \cap c} + \sum_{r \in R_i}M_r^p\ind{v \in A_{\phi} \cap r}} \tag{By def. of $M_r$ and $M_c$}\\
            &\leq 2^{p-1}\sum_{i=0}^{\infty}2^{-ip}\left(\sum_{c \in C_i}\P{v \in A_{\phi} \cap c} + \sum_{r \in R_i}\P{v \in A_{\phi} \cap r} \right) \tag{By def. of $R_i$ and $C_i$}\\
            &\le 2^{p-1}\sum_{i=0}^{\infty}2^{-ip}\cdot \frac{\eta}{\sigma}\left(|C_i| + |R_i|\right)\le  \frac{\eta}{\sigma} 2^{p+2}\sum_{i=0}^{\infty}2^{(1-p)i} = \frac{\eta}{\sigma} \cdot \frac{2^{{p+2}}}{1-2^{1-p}} \le 2^{p+3}\cdot \frac{\eta}{\sigma} \tag{$\sigma$-smoothness and $p \ge 2$}
        \end{align*}
        \begin{figure}[t!]
          \centering
          \begin{subfigure}{0.315\linewidth}
            \centering
            \includegraphics[width=\linewidth]{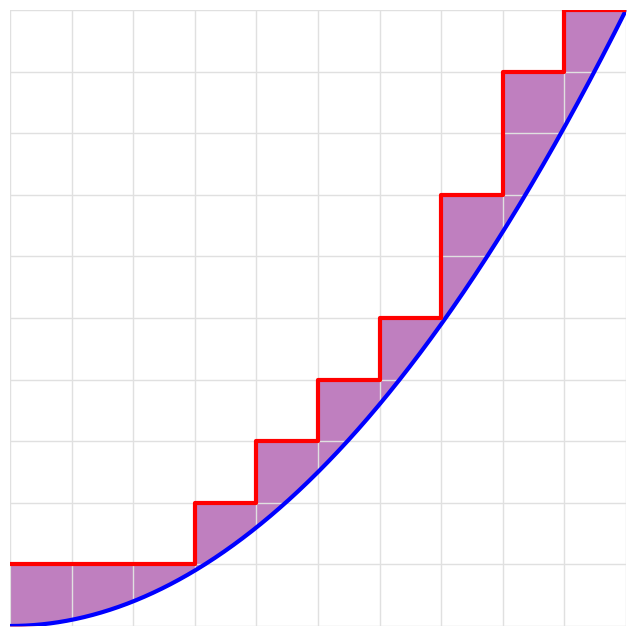}
            \caption{Allocation regions' difference.}\label{fig:approximation_difference}
          \end{subfigure}%
          \hspace*{6pt}
          \begin{subfigure}{0.315\linewidth}
            \centering
            \includegraphics[width=\linewidth]{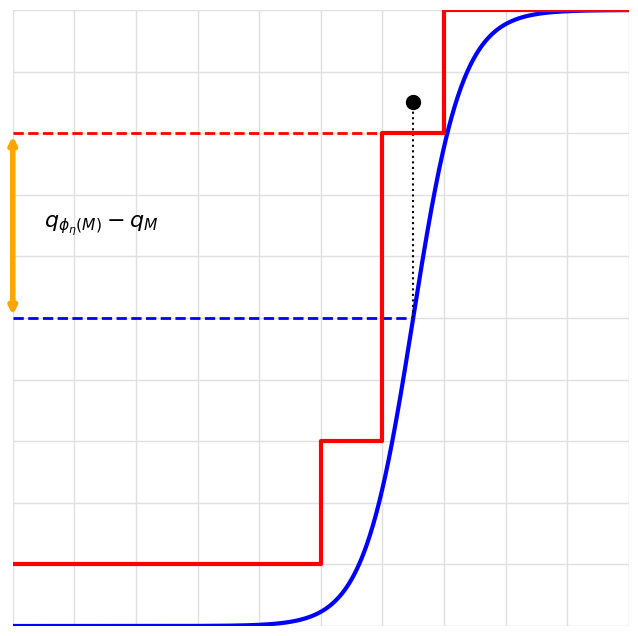}
            \caption{Large buyer's payment gap.}\label{fig:error_buyer}
          \end{subfigure}
          \hspace*{6pt}
          \begin{subfigure}{0.315\linewidth}
            \centering
            \includegraphics[width=\linewidth]{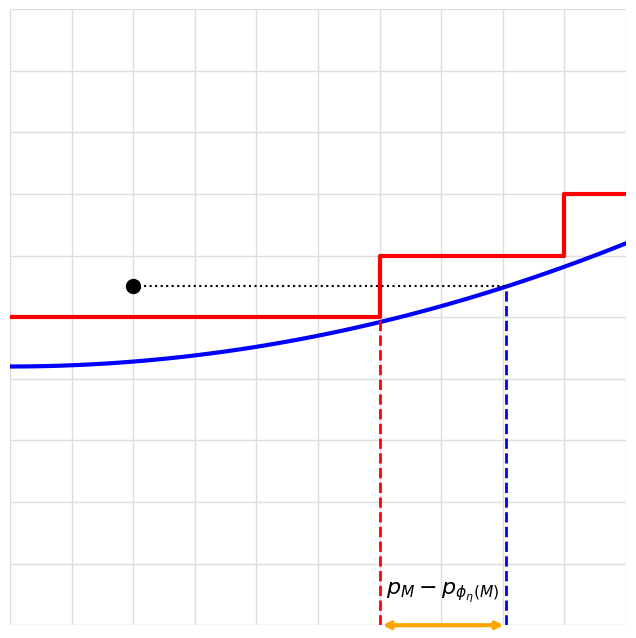}
            \caption{Large seller's payment gap.}\label{fig:seller_error}
          \end{subfigure}
        \caption{Supporting visualizations for \Cref{lemma:restriction,lem:lp_net}. Blue is used for the boundaries of generic mechanisms, while red is for the approximating mechanisms.}
        \label{fig:challenges_lemma3}
        \end{figure}
        
        Note, the core of our argument lies in the second to last inequality and crucially relies on the monotonicity of the allocation region: it is not possible to have more than $2^{i+2}$ rows/columns characterized by a gap in $(2^{-i-1}, 2^{-i}]$! For a formal proof, we defer to \Cref{cl:bound_cl_ro} in Appendix \ref{app:missing_net}.
    \end{proof}

    \subsection{Uniform Convergence}
    \label{subsec:unif_convergence}

        We now prove (fast) uniform convergence on a suitable family of simple mechanisms $\Mnet$. In particular, we consider precision $\eta = \nicefrac{\alpha \sigma}{8}$, so that $\Mnet = \M_{\nicefrac{\alpha \sigma}{8}}$. Given a generic mechanism $M$ and sample $S=\{v_1,v_2,\dots,v_n\}$, we denote with $\prof(M, S)$ the empirical performance of $M$ on $S$: $\prof(M,S) = \nicefrac{1}{n} \sum_{i} \prof(M,v_i).$
        \begin{theorem}
        \label{thm:high_prob_conv_profit}
        Fix any $\acc>0$, $\fail \in (0,1)$. Let $S$ be a $\sigma$-smooth i.i.d. sample of $n$ pairs of valuations, with $v$ drawn independently from the same distribution. Assuming $n\geq \frac{40^5\cdot \log^2(\nicefrac{50}{\alpha\sigma})}{\alpha^2\sigma} + \frac{128\log \nicefrac{2}{\beta}}{\acc^2}$, the following uniform convergence bound holds: 
        \begin{equation*}
            \Psub{S}{\max_{M \in \Mnet} \Bigl\lvert\prof(M,S) - \Esub{v}{\prof(M, v)}\Bigr\rvert \geq \frac{\alpha}4} \leq \frac{\fail}2
        \end{equation*}
        \end{theorem}
        \begin{proof}
            We slightly abuse notation and denote with $\mathcal M_h$ (shorthand of $\M_{2^{-h}}$) the class of $2^{-h}$-simple mechanisms, so that $\app{h}{\cdot}$ is the generic approximating $2^{-h}$-simple mechanism. The \textit{chain} we use to prove the strong uniform convergence bound above is exactly $\{\phi_h(M)\}_{h=0}^H$, for $H=\log_2 \nicefrac{8}{\alpha \sigma}$.

            As a first step, we relate the uniform convergence error in $\Mnet$ to the maximum over a suitable vector of Gaussian variables. As the underlying procedure (symmetrization, then Rademacher Complexity, then Gaussian complexity, and finally a simple telescopic argument) is standard, we defer the proof to Appendix~\ref{app:missing_uniform}.
            \begin{restatable}{claim}{telescope}
            \label{cl:unif_to_telescope}
                Let $\boldsymbol{g} = \{g_i\}_{i=1}^n$ be a vector of independent standard Gaussian variables, then: 
            \begin{align*}
    &\Esub{S}{\max_{M \in \Mnet} \Bigl\lvert\prof(M,S) - \Esub{v}{\prof(M, v)}\Bigr\rvert} \\
    & \leq \sqrt{2\pi}\left(\sum_{h=0}^{H-1}\Esub{S,\textbf{g}}{\max_{M \in \Mnet} \left\lvert \frac{1}{n}\sum_{i=1}^n(\prof(\phi_{h+1}(M),v_i) - \prof(\phi_h(M),v_i)) \cdot g_i\right\rvert}+ \sqrt{\frac{2}{\pi n}}\right)
    \end{align*}
    \end{restatable}

    Given this result, we have to control the expected maximum appearing inside the term above, across all levels of approximation. To do this, we define a target clean event $\clean$, enforcing that, for each level $h \in \{0, \dots, H-1\}$ and for every mechanism $M\in\hat\M$, the following guarantee is satisfied: 
    \begin{equation*}
    \sum_{i=1}^n\left(\prof(\phi_{h+1}(M), v_i)-\prof(\phi_{h}(M), v_i)\right)^2 < 35n\cdot 2^{-h} \cdot \nicefrac{1}{\sigma}
    \end{equation*}

\begin{claim}
\label{cl:clean_event_control} As long as $n\geq \frac{35^3}{\sigma\alpha^2}$, then $\P{\clean} \ge 1-\acc$.
\end{claim}
\begin{proof}[Proof of \Cref{cl:clean_event_control}]
    The bound on the clean event follows by a suitable application of Bernstein’s inequality (see \Cref{prop:bernstein} in Appendix~\ref{app:concentrations} for the exact version of the bound we use). For any sample index $i$, mechanism $M$ and level $h$, we introduce the  random variable $Z_i$ as follows
    \[
        Z_i = (\prof(\phi_{h+1}(M), v_i)-\prof(\phi_h(M), v_i))^2,
        \]
        where, for brevity, we omit the dependence on $M$ and $h$. The expected value of $Z_i$ can be bounded using \Cref{lem:lp_net}, with $p=2$; note, we are crucially relying on the simple observation that $\app{h}{\app{h+1}{M}} = \app{h}{M}$, since each tile in $\cT_{h+1}$ is contained in a bigger one in $\cT_{h}$: $\E{Z_i} \le 34\cdot \nicefrac{2^{-h}}{\sigma}.
        $
        Similarly, we can bound the second moment of $Z_i$ (and thus its variance) by applying \Cref{lem:lp_net} for $p=4$: $\Var{Z_i} \le \E{Z_i^2} \le 130 \cdot \nicefrac{2^{-h}}{\sigma}.$ As a last ingredient to apply Bernstein, we observe that $Z_i - \E{Z_i} \le 4$, as the profit function is upper bounded by $1$ in absolute value. We have: 
    \begin{align*}
    \P{\sum_{i=1}^n Z_i \geq 35n \cdot 2^{-h}\cdot \nicefrac{1}{\sigma}} &\leq \exp\left(-\frac{(35n\cdot 2^{-h}\cdot \nicefrac{1}{\sigma}-\sum_{i=1}^n\E{Z_i})^2}{2 \left(130 n\cdot 2^{-h}\cdot\nicefrac{1}{\sigma}\right)+ \nicefrac{8}{3}\cdot (35n\cdot 2^{-h}\cdot \nicefrac{1}{\sigma}-\sum_{i=1}^n\E{Z_i})}\right) \\ 
    &\leq \exp\left(-\frac{(n\cdot 2^{-h}\cdot \nicefrac{1}{\sigma})^2}{2 \left(130 n\cdot 2^{-h}\cdot\nicefrac{1}{\sigma}\right)+ \nicefrac{8}{3}\cdot (n\cdot 2^{-h}\cdot \nicefrac{1}{\sigma})}\right)\\
    &= \exp\left(-\frac{3n\cdot 2^{-h}\cdot \nicefrac{1}{\sigma}}{788}\right) \le \exp\left(-\frac{3n\cdot \alpha}{6304}\right) \tag{As $2^{-h} \ge \nicefrac{\alpha \sigma}{8}$}
    \end{align*}
    Note, in the second inequality we use that $f(x) = \nicefrac{x^2}{a+bx}$ is monotonically increasing for $x>0$ and positive $a$ and $b$, and we know that $\E{Z_i} \le 34\cdot \nicefrac{2^{-h}}{\sigma}.$ At this point, we union bound over all the levels $h\in \{0, \dots, H-1\}$ and all the mechanisms in $\Mnet$. From \Cref{prop:card_bound}, we know that $\lvert \Mnet \rvert \leq (2e)^{\nicefrac{8}{\acc\sigma}}$. We thus need to satisfy the following inequality, which is satisfied with our lower bound on $n$: 
    $$
    (2e)^{\nicefrac{8}{\acc\sigma}}\cdot \log_2(\nicefrac{8}{\acc\sigma}) \cdot \exp\left(-\frac{3n\cdot \alpha}{6304}\right)\leq \acc.\qedhere
    $$ 
    \end{proof}
    We can go back to studying the right-hand-side term appearing in \Cref{cl:unif_to_telescope}, conditioning on the clean event $\cE$, whose probability we have assessed in \Cref{cl:clean_event_control}.
    
    \begin{claim}[Bounding the links in the chain]
    \label{cl:telescop_conc} The following inequality holds:
    \begin{equation*}
        \Esub{S, \textbf{g}}{\max_{M \in \Mnet} \Bigl\lvert \frac{1}{n}\sum_{i=1}^n(\prof(\phi_{h+1}(M),v_i) - \prof(\phi_h(M),v_i)) \cdot g_i\Bigr\rvert\,\Big|\, \clean} \leq \frac{22}{\sqrt{{\sigma n}}}
    \end{equation*}
    \end{claim}
    \begin{proof}[Proof of \Cref{cl:telescop_conc}]
        Since we are conditioning on $\clean$, we know that the following inequality holds at all levels $h\in \{0, \dots, H-1\}$ and for all mechanisms $M\in \Mnet$:
        \begin{equation}\label{eq:clean_var}
            \frac{1}{n}\sum_{i=1}^n\left(\prof(\phi_{h+1}(M), v_i)-\prof(\phi_{h}(M), v_i)\right)^2 < 35\cdot 2^{-h} \cdot \nicefrac{1}{\sigma}
        \end{equation}
        For any fixed realization of the sample $S$, we observe that the random variables of the form 
        \[
        \frac{1}{n}\sum_{i=1}^n(\prof(\phi_{h+1}(M),v_i) - \prof(\phi_h(M),v_i)) \cdot g_i
        \]
        follow a Gaussian distribution. Conditioning on the clean event and exploiting \Cref{eq:clean_var}:
        \begin{equation}
        \label{eq:clean_var_bound}
            \mathbb{V}\text{ar}_{\boldsymbol{g}}\left(\frac{1}{n}\sum_{i=1}^n(\prof(\phi_{h+1}(M), v_i)-\prof(\phi_{h}(M), v_i)) \cdot g_i \right) \leq \nicefrac{35}{n}\cdot 2^{-h}\cdot \nicefrac{1}{\sigma}.
        \end{equation}

        The bound on the variance allows us to apply a folklore result for Gaussian random variables, namely that for any finite set $G$ of arbitrary Gaussian variables it holds that $\E{\max_{g \in G}|g|} \le 2\sqrt{B\log|G|} $, where $B$ is a uniform bound on the variances (this result is stated precisely as \Cref{prop:max} in Appendix~\ref{app:concentrations}). To apply such inequality, we then need to bound the number of Gaussian variables that appear in the $\max$. At each level, and for each mechanism $M$, we are considering only the approximations $\phi_{h+1}(M)$ and $\phi_h(M)$ so that   $|G| \le |\mathcal M_{h+1} \times \mathcal M_{h}|$. By the bound in \Cref{prop:card_bound}: 
        \[
        |G| \leq \lvert \mathcal M_{h+1}\rvert \cdot \lvert \mathcal M_{h}\rvert \leq \lvert \mathcal M_{h+1}\rvert^2 \leq (2e)^{2 \cdot 2^{h}}.
        \]
        Plugging the above and \Cref{eq:clean_var_bound} into the folklore bound reported in \Cref{prop:max} we get:
        \begin{align*}
        &\Esub{S, \textbf{g}}{\max_{M \in \Mnet } \Bigl\lvert \frac{1}{n}\sum_{i=1}^n(\prof(\phi_{h+1}(M),v_i) - \prof(\phi_h(M),v_i)) \cdot g_i\Bigr\rvert\,\Big|\, \clean}\\
        &\leq 2\sqrt{\nicefrac{35}{n}\cdot 2^{-h}\cdot \nicefrac{1}{\sigma}\cdot 2\cdot 2^h \cdot \log(2e)} \leq 22\sqrt{\nicefrac{1}{\sigma n}}\qedhere
        \end{align*}
    \end{proof}
    Let's now take a step back: in \Cref{cl:unif_to_telescope} we have reduced the problem of controlling the estimation error to handling the sum of the $\max$ of certain random processes. Then, we have introduced a target clean event $\clean$ and argued that (i) the clean event has large probability (\Cref{cl:clean_event_control}) and (ii) conditioning on $\clean$ then the sum appearing in \Cref{cl:unif_to_telescope} is ``easy'' to bound (\Cref{cl:telescop_conc}). Combining these two results via the law of total probability we get the following bound on the \emph{expected} estimation error, its simple proof is deferred to Appendix~\ref{app:missing_uniform}.
    \begin{restatable}{claim}{unifconv}
    \label{cl:unif_conv}
        If $n \geq \frac{35^3}{\sigma\alpha^2}$, then the following uniform convergence bound holds:
    $$
    \Esub{S}{\max_{M \in \hat \M} \Bigl\lvert\prof(M,S) - \Esub{v}{\prof(M, v)}\Bigr\rvert} \leq  \log \left(\frac n{\sigma}\right)\frac{145}{\sqrt{n\sigma}}
    $$
    \end{restatable}

    To conclude the proof of the theorem, we need to get a bound on the uniform convergence error that holds with \emph{high probability}. We achieve this by combining \Cref{cl:unif_conv} with McDiarmid's inequality (we report the exact statement as \Cref{prop:bounded_difference} in Appendix~\ref{app:concentrations}). To this end, we only need to argue that the uniform convergence error $U$ (as a map on the samples $S$) satisfies the \textit{bounded difference} condition. Formally, the function $U$ defined as 
    \(
        U(S) = \max_{M \in \Mnet} \Bigl\lvert\prof(M,S) - \Esub{v}{\prof(M, v)}\Bigr\rvert
    \)
    should respect the property that for any neighboring dataset $S$ and $S'$ (i.e., which differ by a single point), it holds that $|U(S) - U(S')| \le \nicefrac{2}{n}$. To see why this is the case, observe that the expectation $\Esub{v}{\prof(M,v)}$ does not depend on $S$ or $S'$, so that we can thus focus only on the first term inside the absolute value. Considering the vector $\{\prof(M, S)\}_{M\in \Mnet}$, any of its entries may vary by at most $\nicefrac{2}{n}$, since $\prof$ is bounded in absolute value by $1$ and we are taking the average over $n$ valuations. By McDiarmid's inequality, we thus have that, for any positive $t$, the following inequality holds:
    \(
        \P{U(S) \ge \E{U(S)} + t} \le \exp\left(- \frac{nt^2}{2}\right).
    \)
    
    By \Cref{cl:unif_conv}, valid due to our lower bound on $n$, we know that $\E{U(S)} \le \log \left(\frac n{\sigma}\right)\frac{145}{\sqrt{n\sigma}}$, so that:
    \begin{equation}
    \label{eq:inequality_bounded_diff}
     \P{U(S) \geq \nicefrac{\acc}{4}} \leq \P{U(S) - \E{U(S)} \geq \sqrt{\nicefrac{2}{n}\log \nicefrac{2}{\fail}}} \leq \nicefrac{\fail}2,
     \end{equation}
     where the first inequality follows by our assumed lower bound on $n$ and the second follows from McDiarmid's inequality by solving $\exp(\nicefrac{-nt^2}{2}) = \frac{\fail}2$ for $t$.
    \end{proof}
    
    \subsection{An $(\acc, \beta)$-optimal and computationally efficient $\eps$-private algorithm}
    \label{subsec:wrapping}
    
        We have all the ingredients to invoke \Cref{thm:meta}, using as $\Mnet$ the family of $\eta$-simple mechanisms, for $\eta=\nicefrac{\acc\sigma}{8}$. Property (i) follows by \Cref{lemma:restriction}, while (ii) is due to the uniform convergence result in \Cref{thm:high_prob_conv_profit}. Finally, $|\Mnet|$ is bounded in \Cref{prop:card_bound}. All in all, the exponential mechanism on $\Mnet$ which uses as score function the empirical profit verifies the following theorem. 

    \begin{theorem}
    \label{thm:main_profit}
        Fix any $\acc > 0, \beta \in (0, 1)$ and $\eps>0$. Let $S$ be a $\sigma$-smooth i.i.d. sample of size $n$. There exists a polynomial-time randomized learning algorithm which is $(\acc, \beta)$-optimal and $\eps$-differentially private for $\prof$ maximization, provided 
        \[
            n\geq \frac{40^5\cdot \log^2(\nicefrac{50}{\alpha\sigma})}{\alpha^2\sigma} + \frac{128\log \nicefrac{2}{\beta}}{\acc^2} + \frac{256\log(\nicefrac{2}{\beta})}{\acc^2\sigma\eps}\in \Theta\left(\frac{1}{\alpha^2 \sigma}\left[\log^2(\nicefrac{1}{\alpha\sigma})+ \frac{\log \nicefrac{1}{\beta}}{\eps}\right]\right)
        \]
    \end{theorem}
    The only part of the Theorem that we still need to prove concerns how to run the exponential mechanism on the exponentially large $\Mnet$. We present here a computationally efficient sampling routine, deferring the proofs to Appendix~\ref{subsec:missing35}. Before proceeding, recall the definition of the directed graph $G_\eta = (V_\eta, E_\eta)$ as in \Cref{subsec:restriction_profit}: $V_\eta$ is given by the points in the uniform grid of step size $\eta$, while the edges $E_\eta$ are the natural ones from the grid, oriented either from left to right or bottom-up. 
    \begin{definition}
    \label{def:path_grid}
    A path $\pi$ on $G_\eta$ is complete if it starts in $(0, 0)$ and ends in $(1, 1)$. A complete path $\pi$ identifies a mechanism $M_\pi \!\in\! \M_\eta$ of allocation region
    \(
                A_{\pi} \!=\! \{y \in [0,1]^2\!: \!\exists x \in \pi \text{ \normalfont{s.t.} } x_1 \geq y_1, x_2\leq y_2\}. 
    \)
    \end{definition}
    
    From now on, we thus consider only the set of complete paths
    in $G_\eta$, which we denote as $\pathspace$. There is a natural bijection between mechanisms in $\M_{\eta}$ and paths $\pathspace$, so if we consider the exponential mechanism over $\M_{\eta}$, we can equivalently consider a suitable distribution $\mu$ over $\pathspace$. 
    \begin{restatable}{theorem}{sampling}
    \label{theor:sampling}
    Let $S$ be a sample of $n$ valuations, and let $\mu$ be the distribution on $\pathspace$: 
    \begin{equation}
    \label{eq:mu_def}
    \mu(\pi) = \frac{\exp(n\cdot \nicefrac{\eps}{4} \cdot \prof(M_\pi, S))}{\sum_{\pi'\in \pathspace}\exp(n \cdot \nicefrac{\eps}{4} \cdot \prof(M_{\pi'}, S))}\quad \forall \, \pi\in \pathspace
    \end{equation}
    There exist edge probabilities $\{q_e\}_{e\in E_\eta}$, computable in time $O(\nicefrac{n}{\eta^2})$, such that:
    
    \begin{enumerate}[(i)]
        \item $q_e\in[0, 1]\ \forall\, e\in E_\eta$.
        \item $\sum_{e\in \neighedge{u}}q_e = 1\ \forall \, u\in V_\eta$, where $\neighedge{u}$ is the set of outgoing edges from node $u$ in $E_\eta$.
        \item $\prod_{e\in \pi}q_e = \mu(\pi)\ \forall\, \pi\in \pathspace$, almost surely with respect to the sampling of $S$.
    \end{enumerate}
    \end{restatable}
    The theorem implies we can \emph{simulate} the distribution $\mu$ supported over paths with a random walk on $V_\eta$, starting at $(0, 0)$ and ending in $(1,1)$, which builds the sampled path ``one edge at a time''. Note, the number of weights to keep in memory is then polynomial (as $|E_{\eta}| \in O(\nicefrac{1}{\eta^2})$). 
    
    We conclude the section by explaining how to compute these $q_e$, in polynomial time. As a first step, we introduce a notion of edge-weights $w_e$, in such a way that the profit of a generic mechanism on a sample set can be decomposed as the sum of the logarithm of the weights of its edges: $\prof(M_{\pi},S) \propto \sum_{e\in \pi} \log w_e$.
    \begin{definition}[Weighting the edges in $G_\eta$] 
    \label{def:weight_edges}
    Consider an edge $e$ in $G_\eta$. We define its interest region $A_e$:
    \begin{itemize}
        \item $A_e = [0, {k}{\eta}]\times [{j}{\eta}, {(j+1)}{\eta}]$ if $e$ is a vertical edge connecting $({k}{\eta}, {j}{\eta})$ to $({k}{\eta}, {(j+1)}{\eta})$.
        \item $A_e = [{k}{\eta}, {(k+1)}{\eta}]\times [{j}{\eta}, 1]$ if $e$ is a horizontal edge connecting $({k}{\eta}, {j}{\eta})$ to $({(k+1)}{\eta}, {j}{\eta})$.
    \end{itemize}
    Let $S = \{v_1, \dots, v_n\}$ be a sample on $n$ valuations. We assign weights $w_e$ to the edges of $G_\eta$: 

    \begin{equation*}
        w_e = \begin{cases}
            \exp\left(-\nicefrac{\eps}{4}\cdot{k}{\eta}\sum_{i=1}^n\ind{v_i\in A_e}\right)\ &\text{if $e$ connects $({k}{\eta}, {j}{\eta})$ to $({k}{\eta}, {(j+1)}{\eta})$.} \\
            \exp\left(\nicefrac{\eps}{4}\cdot{j}{\eta}\sum_{i=1}^n\ind{v_i\in A_e}\right) \ &\text{if $e$ connects $({k}{\eta}, {j}{\eta})$ to $({(k+1)}{\eta}, {j}{\eta})$.}
            \end{cases}
    \end{equation*}
    \end{definition}
    The second step of the construction consists in deriving the right ``transition probabilities''.
    Since each step in a random walk on $G_\eta$ presents a choice between at most two possible next nodes, or equivalently two edges, we define these transitions by balancing the immediate edge weight with the $\prof$ of the paths that follow. This leads to our second intuition: weighting nodes recursively by summing the $\prof$ contributions of all possible \textit{future} paths stemming from them.
    
    \begin{definition}[Weighting the nodes in $G_\eta$]
    \label{def:weight_nodes}
    Let $u$ be a generic vertex of $G_\eta$ and let $\neigh{u}$ be its out-neighborhood. Let $w_{(1, 1)}=1$ and recursively define the weight of $u$ as $w_u = \sum_{s\in \neigh{u}}w_{(u, s)}w_s$.
    \end{definition}

        \begin{figure}[t!]
          \centering
          \begin{subfigure}{0.27\linewidth}
            \centering
            \includegraphics[width=\linewidth]{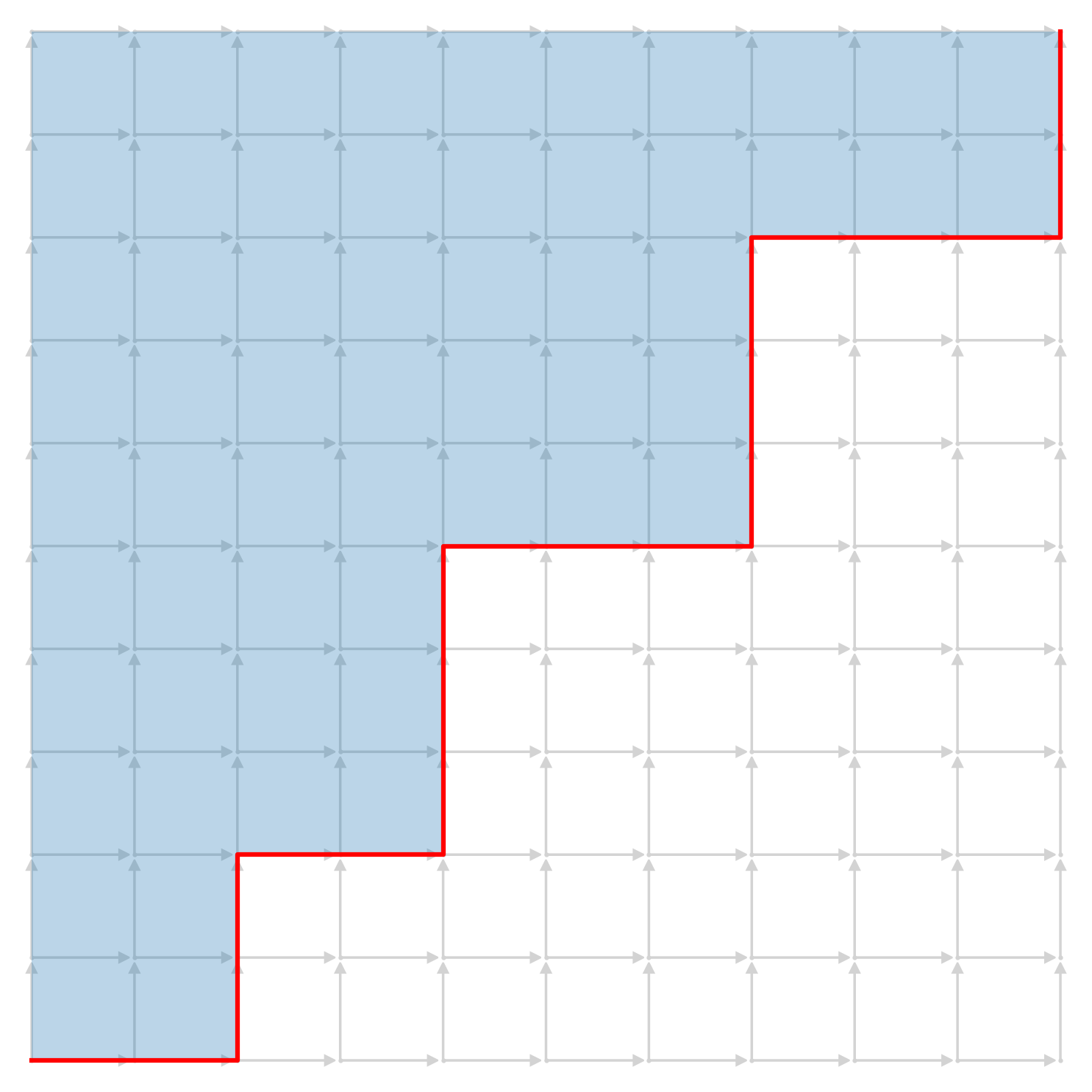}
          \end{subfigure}%
          \hspace*{26pt}
          \begin{subfigure}{0.27\linewidth}
            \centering
            \includegraphics[width=\linewidth]{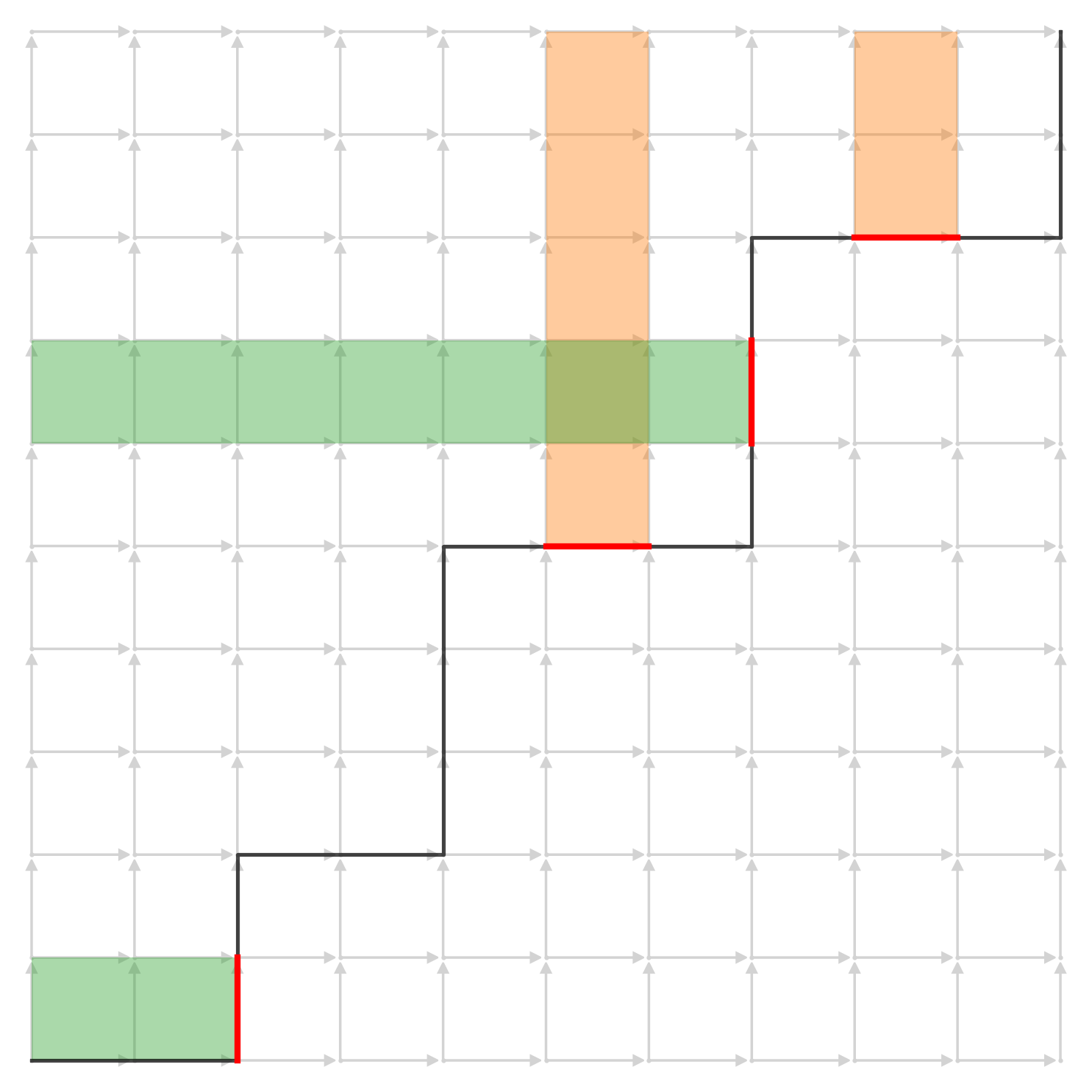}
          \end{subfigure}
        \caption{To the left, a path (in red) and its induced allocation region (in blue) are illustrated. To the right, four interest regions (as defined in \Cref{def:weight_edges}) are shown for four different edges: two vertical and two horizontal.}
        \label{fig:graph_algorithm}
        \end{figure}

    Due to the recursive definition, we stress that $\{w_u\}_{u\in V_\eta}$ can be computed by dynamic programming in polynomial time. At any point in the random walk, given $\{w_e\}_{e\in E_\eta}$ and $\{w_u\}_{u\in V_\eta}$, we can now account for both the immediate contribution of the next edge to $\prof$ and the aggregated contribution of all potential future paths. The natural idea is then multiplying the two and normalizing to get probabilities. More specifically, for any node $u\neq (1, 1)$ in $V_\eta$, if $s\notin \neigh{u}$, then we define $q_{(u, s)} = 0$; if $\neigh{u} = \{s\}$, then we define $q_{(u, s)} = 1$; if $\neigh{u} = \{s, r\}$: 

     \begin{equation}
        \label{eq:prob_definition}
        q_{(u, s)} = \frac{w_{(u, s)}w_s}{w_{(u, s)}w_s + w_{(u, r)}w_r} \quad \quad q_{(u, r)} = \frac{w_{(u, r)}w_r}{w_{(u, s)}w_s + w_{(u, r)}w_r}.
    \end{equation}
    This concludes the presentation of the efficient sampling routine. We conclude this section with a remark on the implementability of our approach. Namely, we comment on what the algorithm can do if it does not know the smoothness parameter $\sigma$ (or a constant upper bound of it).

    \begin{remark}[On not knowing $\sigma$]
        \label{rmk:profit}
            We can run the exponential mechanism on $\M_\eta$ with $\eta = \nicefrac{\alpha}{8}$, irrespective of $\sigma$. With minor adjustments to the previous analysis, it can be shown that $\tilde \Theta\left(\nicefrac{\sigma}{\alpha^2\eps}\right)$ samples are sufficient to (privately) obtain a $(\nicefrac{\alpha}{\sigma},\beta)$-optimal mechanism. The latter bound is of the same order as that in \Cref{thm:main_profit}, modulo rescaling accordingly the desired precision parameter.
    \end{remark}
    
\section{Gain-From-Trade Maximization}

    In this section, we study gain from trade maximization, with the goal of learning the best fixed price mechanism $M_p$. First, in \Cref{subsec:lower_gft}, we prove that private learning is unattainable for general distributions. Then, in \Cref{subsec:upper_gft}, we show how to instantiate our Meta-Theorem (\Cref{thm:meta}). 

    \subsection{Private Learning is Impossible for General Distributions}
    \label{subsec:lower_gft}

    In this section, we argue that, in general, it is not possible to privately learn the best fixed-price mechanism for bilateral trade. We report the statement of the impossibility result, and a proof sketch, while the complete argument is deferred to Appendix~\ref{app:lower_gft}.

        \begin{restatable}{theorem}{gftimpossible}
        \label{thm:impossibility_gft}
        For any learning algorithm $\cA$ for gain from trade maximization, precision parameter $\alpha \in (0,\nicefrac 1{80})$, privacy parameter $\varepsilon>0$ and number of samples $N$, there exists a hard distribution $\cD$ for which the learning algorithm $\cA$ on $N$ i.i.d. samples cannot be $\eps$-differentially private and satisfy:  
            \[
                \sup_{p \in [0,1]} \E{\gft(p,v)} \le \E{\gft(\A(S),v)} + \alpha,
            \]
            where the expectation is with respect to $S$ and $v$ drawn from $\cD$, and the randomization of  $\cA$. 
        \end{restatable}
        \begin{proof}[Proof Sketch]
             We construct a ``hard'' family of distributions, indexed by the uniform grid $\Gamma$ over the interval $[\nicefrac{9}{20},\nicefrac{11}{20}]$, with step-size $\xi$, to set later. To each such $\gamma_i \in \Gamma$, we associate a distribution $\cD_i$ with the following two properties: (i) the optimal price given $\cD_i$ is $\gamma_i$ and (ii) any price ``far'' from $\gamma_i$ results in a gain from trade that is suboptimal by more than $\alpha$. 

             Consider now any learning algorithm $\cA$. By point (ii), for each $\gamma_i \in \Gamma$ there should exist a sample $S_i$ such that $\cA(S_i)$ posts a price close to $\gamma_i$, let's say in some neighborhood $I_i$ of it, but far from \emph{all the other} $\gamma_j \in \Gamma$, so that $I_i \cap I_j = \emptyset.$  For the generic pair $\gamma_i, \gamma_j \in \Gamma$, the corresponding samples $S_i$ and $S_j$ may be different, but \emph{for sure}, they differ by at most $N$ points. Therefore, we can construct a sequence of sets $S_i=S^0_{i,j}, S^1_{i,j}, \dots, S^N_{i,j}=S_j$ such that $S^k_{i,j}$ and $S^{k+1}_{i,j}$ differ by at most one point. We can then repeatedly apply the definition of $\eps$-differentially private algorithm on $\A$:
            \begin{align}
                \Psub{\A}{\A(S_i) \in I_j} &\ge e^{-\eps} \Psub{\A}{\A(S^1_{i,j}) \in I_j}\tag{$S_i$ and $S^1_{i,j}$ differ in one point}\\
                &\ge e^{-\eps}\left[e^{-\eps} \Psub{\A}{\A(S^2_{i,j}) \in I_j}\right]  \tag{$S^1_{i,j}$ and $S^2_{i,j}$ differ in one point}\\
                &\ge e^{-N\eps}\Psub{\A}{\A(S_{j}) \in I_j}\tag{iterating $N$ times}
            \end{align}

            Now, in order to be $\alpha$-optimal, the algorithm $\cA$ should respect that $\P{\cA(S_i) \in I_i}$ is greater than some constant, let's say $\nicefrac 13$, for all $i$. Since the $I_i$ are all disjoint, we have
            \[                 \Psub{\A}{\A(S_i) \notin I_i} \ge \sum_{j\neq i} \Psub{\A}{\A(S_i) \in I_j} \ge \left(\frac{1}{\xi} - 2\right)e^{-N\eps} \Psub{\A}{\A(S_{i}) \in I_i} \ge \frac{1}{3}\left(\frac{1}{\xi} - 2\right)e^{-N\eps}.
            \]
            We can take $\xi$ as small as possible, while $\eps$ and $N$ are fixed, thus arriving to the contradiction that the mechanism proposed on sample $S_i$ prefers posting away from $\gamma_i$, thus violating (i) and (ii). 
        \end{proof}

    \subsection{A Fixed Grid over Prices} \label{subsec:upper_gft}

    Given accuracy $\alpha$ and the smoothness parameter $\sigma$, we restrict our attention to the family of mechanisms corresponding to the price-grid $\hat\cG=\{ \nicefrac{\alpha\sigma\cdot i}{4}, \, i=0, \dots, \nicefrac{4}{\alpha\sigma}\}$. 
    Under the smoothness assumption, the function $p \to \E{\gft(M_p,v)}$ is Lipschitz
    of parameter $\nicefrac{1}{\sigma}$ (Lemma 1 of \citet{Cesa-BianchiCCF24jmlr}), so that we have the following discretization bound. 
    \begin{lemma}
    \label{lem:discretization_gft}
        Consider any accuracy $\alpha$ and $\sigma$-smooth random variable $v$, we have the following inequality:
        \(
                    \max_{p \in [0,1]} \Esub{v}{\gft(M_p,v)} - \max_{p \in \hat \cG} \Esub{v}{\gft(M_p,v)} \le\frac{\acc}4.    
        \)
    \end{lemma}

    Concerning uniform convergence, it follows immediately from a simple union bound, which we defer to Appendix~\ref{app:missing42}. We can then plug the two preliminary Lemmata into \Cref{thm:meta}. 

    \begin{restatable}{lemma}{uniformgft}
    \label{lem:uniform_gft}
    Consider any accuracy $\alpha$, failure probability $\beta$, and $\sigma$-smooth random variable $v$. Denote with $S$ an i.i.d. sample from $\cD$, then    
        \(
            \max_{p \in \hat \cG}|\Esub{v}{\gft(M_p,v)} - \gft(M_p,S)| \le \frac{\acc}4,
        \)
    with probability at least $1-\nicefrac{\beta}{2}$, as long as $|S| \ge \nicefrac{8}{\alpha^2}\log\left(\nicefrac{32}{\alpha\beta\sigma}\right)$.
    \end{restatable}
    \begin{theorem}
    \label{thm:main_gft}
        Fix any $\acc > 0, \beta \in (0, 1)$ and $\eps>0$. Let $S$ be a $\sigma$-smooth i.i.d. sample of size $n$. Then the exponential mechanism on $\hat \cG$ is $(\acc, \beta)$-optimal and $\eps$-differentially private for $\gft$ maximization, provided $n\geq \frac{16}{\alpha}\left(\frac{1}{\alpha}+\frac{1}{\eps}\right)\log\left(\frac{32}{\alpha\beta\sigma}\right)$.
    \end{theorem}

    \begin{remark}
    \label{rmk:gft}
        In general, we do not need to know the smoothness parameter $\sigma$ exactly, as any estimate $\hat \sigma \le \sigma$ can be used. Since the sample-complexity dependency on the smoothness parameter is only logarithmic in \Cref{thm:main_gft}, then any $\hat \sigma$ such that $\nicefrac{ \hat \sigma}{\sigma}$ is polynomial in $\alpha$ and $\beta$ only affects the sample complexity bound by a constant factor. Alternatively, we can still recover roughly the same rate, by considering a price grid of step-size $\nicefrac{\alpha}{4}$, independent of $\sigma$. This yields an extra $\nicefrac{1}{\sigma}$ factor in the bound of \Cref{lem:discretization_gft}, so that the same error can be accounted for in \Cref{lem:uniform_gft} using $\tilde{\Theta}({\nicefrac{\sigma^2}{\alpha^2}})$ samples. Plugging this in the proof of \Cref{thm:meta} implies that we can get an $(\nicefrac{\alpha}{\sigma},\beta)$-optimal mechanism using $\tilde{\Theta}(\nicefrac{\sigma^2}{\alpha^2} + \nicefrac{\sigma}{\alpha \eps})$ samples (same order as that in \Cref{thm:main_gft}, modulo rescaling accordingly the desired precision parameter).  
    \end{remark}

\section*{Acknowledgments}

    The work of SDG, FF, and SL was supported in part by the MUR PRIN grant 2022EKNE5K (Learning in Markets and Society), while CS was supported by a Google Research Award. S.D.G. was also supported by the Institute for Complex Systems (Italian National Research Council).

\bibliographystyle{plainnat}
\bibliography{bibliography}

\newpage
\appendix
\section{Concentration Bounds}
\label{app:concentrations}

    In this Appendix, we report for completeness the exact versions of some famous probabilistic inequalities used in the main body. First, we report a well-known result relating the maximum absolute values of $N$ Gaussian variables with their variance; its simple proof can be found, for instance, in  Lemma 8 of \citet{DiGregorioDFS25}.
    \begin{proposition}
        \label{prop:max}
        Let $g_1\ldots g_N$ be Gaussian random variables, with null expectation and variance uniformly bounded by $k$. Then, the following inequality holds:
        \[
            \E{\max_{i \in [N]}|g_i|}\leq 2\sqrt{k\log N}.
        \]
    \end{proposition}

    The version of the Chernoff–Hoeffding bound that we use in Appendix \ref{app:uniform_convergence} comes from Theorem 1.1. of \citet{Dubhashi_Panconesi_2009}.    

    \begin{proposition}[Chernoff–Hoeffding bounds]
    \label{prop:chernoff}
        Let $X = \sum_{i \in [n]} X_i$, where $X_i$, $i \in [n]$ are independently distributed in $[0, 1]$. Then,
        \begin{itemize}
            \item For all $t > 0$,
            \[
                \P{X > \E{X}+t},\P{X < \E{X}-t} \le e^{-\nicefrac{2t^2}n}.
            \]
            \item For $\eps$ > 0,
            \begin{align*}
                \P{X>(1+\eps)\E{X}} \le \exp\left({-\frac{\eps^2}{3}\E{X}}\right),\\
                \P{X<(1-\eps)\E{X}} \le \exp\left({-\frac{\eps^2}{2}\E{X}}\right).
            \end{align*}
            \item If $t > 2e\E{X}$, then
            \[
                \P{X>t} \le 2^{-t}.
            \]
        \end{itemize}
    \end{proposition}
    
    We also use Bernstein's inequality, see e.g., Theorem 1.2 of \citet{Dubhashi_Panconesi_2009}. To be more adherent to the source we denote the variance with $\sigma$, thus overloading locally the symbol for the smoothness parameter. 
    \begin{proposition}[Bernstein's inequality]
        \label{prop:bernstein}
        Let the random variables $X_1, X_2, \dots, X_n$ be independent with $X_i - \E{X_i}<b$ for each $i=1, 2,\dots, n$. Let $X=\sum_i X_i$, and let $\sigma^2=\sum_i \sigma_i^2$ be the variance of $X$. Then, for any $t>0$, 
        \[
            \P{X > \E{X} + t} \le \exp\left( -\frac{t^2}{2\sigma^2+\tfrac{2}{3}bt}\right).
        \]
    \end{proposition}
    Finally, we state the bounded differences inequality (also known as McDiarmid's Inequality). First, we introduce the bounded difference property, and then the inequality (respectively Definition 5.6 and Corollary 5.2 in \citet{Dubhashi_Panconesi_2009})
    
    \begin{definition}[Bounded differences condition]
        A function $f(x_1, \dots, x_n)$ satisfies the bounded differences condition with constants $d_i$, $i=1, \dots, n$, if 
        \[
            |f(a) - f(a')| \le d_i,
        \]
        whenever $a$ and $a'$ differ in just the $i^{\textnormal{th}}$ coordinate, $i=1, \dots, n$.
    \end{definition}

    \begin{proposition}[Method of bounded differences]
        \label{prop:bounded_difference}
        If $f$ satisfies the bounded differences property with constants $d_i$,  $i=1, \dots, n$, and $X_1, X_2, \dots, X_n$ are independent random variables, then for any $t>0$,
        \[
            \P{f(X_1, \dots, X_n) > \E{f(X_1, \dots, X_n)} + t} \le \exp\left(- \frac{2t^2}{d}\right),
        \]
        where $d = \sum_i d_i^2$.
    \end{proposition}
\section{Uniform Convergence for Profit Maximization with Smooth Distributions}
    \label{app:uniform_convergence}
     As a consequence of the previous result, one may use a bracketing argument~\citep{bracketing_annals} to extend uniform convergence to the more general class $\M$. The idea is that, for every mechanism $M\in \M$, we can find two functions stemming from $\M_\eta$ that are very close in expectation and that upper bound and lower bound the $\prof$ for $M$.

    \begin{theorem}
    \label{thm:generalunifconv}
    
    Fix any $\acc>0$, $\fail \in (0,1)$. Consider a $\sigma$-smooth distribution, an i.i.d. sample $S$ of $n$ valuations, and with $v$ a fresh sample from it. Assuming $n\geq100(\frac{40^5\cdot \log^4(\nicefrac{50}{\alpha\sigma})}{\alpha^2\sigma} + \frac{\log^2(\nicefrac{8}{\alpha\sigma})\cdot 128\log \nicefrac{2}{\beta}}{\acc^2})$, the following uniform convergence bound holds: 
    \begin{equation*}
        \Psub{S}{\sup_{M \in \M} \Bigl\lvert\prof(M,S) - \Esub{v}{\prof(M, v)}\Bigr\rvert \geq \acc} \leq \fail
    \end{equation*}
    \end{theorem}
    To prove the above, we need additional notation. For $M\in \M$ and $\eta > 0$, let $\phi_\eta^{+}(M)$ be the mechanism whose allocation region is characterized by $A_{\phi_\eta(M)}\cup K_{M, \eta}$, in the notation\footnote{Here we slightly abuse notation, since here we consider $K_{M, \eta}$ as the union of the tiles in $\cT_\eta$ intersecting $\partial A_M$, not the collection of those tiles.} of Appendix \ref{app:missing_net}. Crucially, $\phi_\eta^{+}(M)\in \M_\eta$. Moreover, let $U^\eta_M(v) = \prof(\phi_\eta(M), v)+\ind{v\in K_{M, \eta}}$ and $L^\eta_M(v) = \prof(\phi_\eta^+(M), v)-\ind{v\in K_{M, \eta}}$, with $\hat U_M^\eta$ and $\hat L_M^\eta$ their empirical averages on the sample $S$.
    \begin{proof}
    As a first step, observe that, for any mechanism $M\in\M$ and valuation $v\in [0, 1]^2$, by construction we have $U^{\eta}_M(v) \geq \prof(M, v) \ge L^{\eta}_M(v)$. This follows by a reasoning similar to the one in the proof of \Cref{lemma:restriction} (the inner mechanism achieves higher $\prof$). Set $\eta=\nicefrac{\acc\sigma}{8}$, as in \Cref{subsec:unif_convergence}.

    We first decompose the uniform convergence error, removing the absolute value from the analysis. We thus focus on one direction of the target inequality. We have the following chain of inequalities, which is a consequence of the definition of $L^{\eta}_M$ and $U^{\eta}_M$ and the triangle inequality on the supremum:
    \begin{align*}
        \sup_{M\in \M}\prof(M, S)-&\Esub{v}{\prof(M, v)}\\
        &\leq \sup_{M\in \M}\hat U^{\eta}_M -\Esub{v}{L^{\eta}_M(v)} = \sup_{M\in \M}\hat U^{\eta}_M-\Esub{v}{L^{\eta}_M(v)} \pm \Esub{v}{U^{\eta}_M(v)} \\
        & \leq \underbrace{\sup_{M\in \M}\prof(\phi_{\eta}(M), S)-\Esub{v}{\prof(\phi_{\eta}(M), v)}}_{(A)} +\underbrace{\sup_{M\in \M}\frac{1}{n}\sum_{i=1}^n \ind{v_i\in K_{M, {\eta}}}}_{(B)}\\ &\quad \, + 3\underbrace{\sup_{M\in \M}\Psub{v}{v\in K_{M, {\eta}}}}_{(C)}
        + \underbrace{\sup_{M\in \M}\Esub{v}{\prof(\phi_{\eta}(M), v)-\prof(\phi^+_{\eta}(M), v)}}_{(D)}
    \end{align*}
    We address the above terms separately. The (A) term can be upper bounded by $\nicefrac{\acc}{4}$ with probability $\nicefrac{\fail}{2}$, given the lower bound on $n$ for the uniform convergence bound on $\M_{\nicefrac{\acc\sigma}{8}}$, \Cref{thm:high_prob_conv_profit}. We call this lower bound $n_0$. Term (C) can be bounded by $\nicefrac{\acc}{4}$, using \Cref{cl:cardinality}.

    It remains to bound term (B) and (D). To bound (B), we use the version of the multiplicative Chernoff bound from \Cref{app:concentrations}, \Cref{prop:chernoff}, since the term inside the supremum is just the average of $n$ Bernoulli variables with mean bounded by $\nicefrac{\acc}{4}$, again by \Cref{cl:cardinality}. Fix any $M\in \M$, then the following holds: 
    $$
    \Psub{S}{\frac{1}{n}\sum_{i=1}^n\ind{v_i\in K_{M, {\eta}}}> 2\acc}\leq 2^{-2n\acc} 
    $$

    Notice that by construction of our approximation in \Cref{def:approximations_prof}, there is a bijection between $\{K_{M, \eta}\}_{M\in \M}$ and $\M_{\eta}$, meaning that $\lvert \{K_{M, \eta}\}_{M\in \M}\rvert = \lvert \M_{\eta}\rvert$. Using \Cref{prop:card_bound} and union bounding, we then get: 
    $$
    \Psub{S}{\max_{M\in \M}\frac{1}{n}\sum_{i=1}^n\ind{v_i\in K_{M, \eta}}> 2\acc} \leq 2^{-2n\acc} \cdot (2e)^{\nicefrac{8}{\acc\sigma}}, 
    $$
    which is upper bounded by $\nicefrac{\fail}{2}$ given $n\geq n_0$.

    \begin{figure}[t!]
        \centering
        \includegraphics[width=0.35\linewidth]{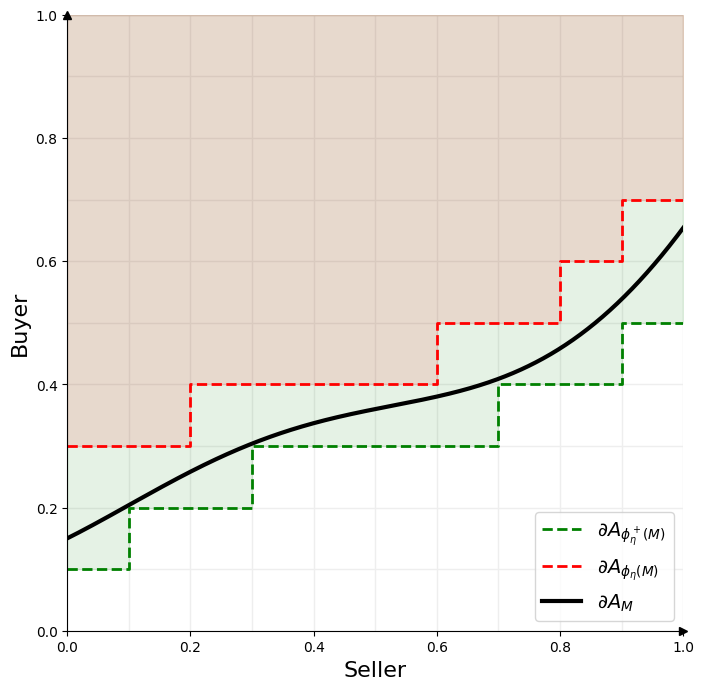}
        \caption{$\phi_{\eta}(M)$ and $\phi^+_\eta(M)$ for a specific mechanism $M$, for $\eta = \nicefrac{1}{10}$.}
        \label{fig:brackets}
    \end{figure}

    Term (D) is instead bounded by considering the same reasoning used at the end of the proof of \Cref{lem:lp_net}, in order to obtain a bound for $p=1$. Specifically, a slight adaptation of \Cref{cl:bound_cl_ro} gives that $C_i$ and $R_i$, in this case defined for the difference between the pricing functions of the inner and outer approximations, have at most $2^{i+3}$ elements each. However, with this value of $p$, the terms summed over the series are constant, so we cannot get any bound from that approach. What we can do is instead taking the sum until $i=\log_2(\nicefrac{1}{\eta})$, so that the remainder error is bounded by $\eta$ and we have the following: 

    \begin{equation*}
        \E{\lvert \prof(\phi_\eta(M), v)-\prof(\phi_\eta^+(M), v)\rvert}\leq \eta + \nicefrac{\eta}{\sigma}\cdot 32\cdot \log_2(\nicefrac{1}{\eta}) + \nicefrac{2\eta}{\sigma}\leq 7\acc\cdot \log(\nicefrac{8}{\acc\sigma})
    \end{equation*}

    Putting everything together and by a union bound with the two events controlling (A) and (B), we get, provided $n\geq n_0$: 
    $$
    \Psub{S}{\sup_{M\in \M}\prof(M, S)-\Esub{v}{\prof(M, v)} > 10\acc\cdot \log(\nicefrac{8}{\acc\sigma})}\leq \fail
    $$
    We then perform the same analysis for the other side of the inequality, observing that we do not need any additional high probability statements. The claimed result follows by rescaling $\acc$, and thus $n_0$, to remove the $10\log(\nicefrac{8}{\acc\sigma})$ factor.

    \end{proof}

\section{Missing Proofs}
\label{app:missing}

\subsection[Missing Proofs from Section 2.1]{Missing Proof from \Cref{sec:preliminaries}}
\label{app:preliminaries}

\fixed*
    \begin{proof}
        Consider any DSIC and IR mechanism $M$ that enforces budget balance. The point $(0,1)$ belongs to the allocation region $A_M$ of $M$ (otherwise $M$ never trades, with zero $\gft$), and results in payment $\underline q$ for the buyer and $\overline p$ for the seller, with $\underline q \ge \overline p$ (budget balance). Consider any $p \in [\overline p,\underline q]$ and the corresponding fixed price mechanism $M_p$. By monotonicity, the allocation region $A_M$ is contained in that of $M_p$, so that the latter performs at least the same trades as $M$ (and thus achieves at least the same $\gft$), while still enforcing budget balance. 
    \end{proof}

\subsection[Missing Proofs from Section 2.2]{Missing Proof from \Cref{subsec:meta}}
\label{app:meta}

\meta*
    \begin{proof}
        We consider the exponential mechanism described above, which we denote with $\cA$ for simplicity. $\cA$ uses $f(\cdot, S)$ as the score function, and as the supporting distribution, the uniform one over $\Mnet$. 

        For any set of samples $S$, denote with $\opt_S$ the $f$-value of the best $M \in \Mnet$ over $S$, and with $M_S^\star$ the associated best mechanism. We introduce the first clean event, namely 
        \[
            \clean_{\text{EXP}} = \{\opt_S - f(\cA(S), S) \le \nicefrac{\acc}4\}.
        \]
        For simplicity, for any set of samples $S$, we denote with $\M_{\text{sub}}$ the mechanisms in $\Mnet$ whose $f$-value over the samples $S$ is more than an additive $\nicefrac{\acc}4$ term away from $\opt_S$. We have the following chain of inequalities:
        \begin{align*}
            \P{\clean_{\text{EXP}}^C} &= \frac{\sum_{M \in \M_{\text{sub}}}\exp\left(\varepsilon\cdot \frac{f(M, S)}{\nicefrac 4n}\right)}{\sum_{M \in \Mnet}\exp\left(\varepsilon\cdot \frac{f(M, S)}{\nicefrac 4n}\right)} \tag{Design of $\cA$}\\
            &\le |\M_{\text{sub}}|\frac{\exp\left(\varepsilon\cdot \frac{\opt_S - \nicefrac{\acc}{4}}{\nicefrac 4n}\right)}{\sum_{M \in \Mnet}\exp\left(\varepsilon\cdot \frac{f(M, S)}{\nicefrac 4n}\right)} \tag{By definition of $\M_{\text{sub}}$}
            \\
            &\le |\Mnet|\frac{\exp\left(\varepsilon\cdot \frac{\opt_S - \nicefrac{\acc}{4}}{\nicefrac 4n}\right)}{\exp\left(\varepsilon\cdot \frac{\opt_S}{\nicefrac 4n}\right)} \tag{By definition of $\opt_S$ and because $\M_{\text{sub}} \subseteq \Mnet$}\\
            &= |\Mnet| e^{-\varepsilon\cdot \frac{{\acc}n}{16}} \le \frac{\fail}2,
        \end{align*}
        where the last inequality follows by our assumption on the number of samples $n.$ 

        We move our attention to the second clean event of our analysis, namely that the uniform convergence property (ii) is indeed realized. Let $\cE$ be given by the intersection of the two clean events. By a simple union bound, we have that $\cE$ has probability at least $1-\fail$. Assuming that $\cE$ is realized, we can prove that $\cA$ is $(\acc,\fail)$-optimal; the privacy guarantee is derived automatically by the properties of the exponential mechanism. We have the following chain of inequalities: 
        \begin{align*}
            \sup_{M \in \M} \Esub{v}{f(M,v)} &\le \max_{M \in \Mnet} \Esub{v}{f(M,v)} + \frac{\acc}4 \tag{By property (i)}\\
            &\le \opt_S + \frac{1}2\acc \tag{By property (ii) and optimality of $\opt_S$}\\
            &\le f(\cA(S), S) + \frac{3}4\acc \tag{By the first clean event}\\
            &\le \Esub{v}{f(\cA(S), v)} + \acc. \tag{By property (ii)}
        \end{align*}
        The above chain of inequalities is verified when the two clean events are realized, which happens with probability at least $1-\fail$. This concludes the proof.
    \end{proof}

\subsection{Missing Proofs from \Cref{subsec:lower_profit}}
\label{app:missing_lower_profit}
\lbprofits*
\begin{proof}
    In the following, let $v_i (t) = (t\xi_i, \nicefrac{1}{2}+\nicefrac{t}{2})$, for any $t\in [0, 1]$. Notice that $t\sim \text{Unif~[0, 1]}$ implies $v_i(t)\sim \cD_i$. Correspondingly, for any $t\in [0, 1]$, let $f_i(t)=\nicefrac{1}{2}+\nicefrac{t}{2}-t\xi_i$ be the maximum $\prof$ achievable from the valuation $v_i(t)$, i.e. the one achieved by $M_i$. Recall that the maximum profit achievable for $\cD_i$ is thus $\opt_i = \int_0^1 \nicefrac{1}{2}+\nicefrac{t}{2}-t\xi_i dt = \nicefrac{3}{4}-\nicefrac{\xi_i}{2}$. 
    
    We proceed by cases, depending on the relation between $i$ and $j$. Both for $i>j$ and for $i<j$, we would like to show that the following upper bound holds: 
    \begin{equation}
    \label{eq:prof_up_lower_bound}
    \Esub{v\sim \cD_j}{\prof(M, v)}\leq \frac{2}{3} 
    \end{equation}
    
    First, assume $i<j$. Consider the set $B_i = \{t\in[0, 1]: v_i(t)\in A_M\}$. A consequence of near-optimality on $\cD_i$ is that the Lebesgue measure of $B_i$ is strictly lower bounded by $1-2\alpha'$. To see why this is the case, assume by contradiction that this is not the case. Notice that $f_i$ is increasing in $t$ if $\xi_i<\nicefrac{1}{2}$, so that, among all subsets $B$ of $[0, 1]$ of measure at most $1-2\alpha'$, the one maximizing $\int_B f_i(t)dt$ is $B = [2\alpha', 1]$. This means that: 
    \begin{align*}
    \Esub{v\sim \cD_i}{\prof(M, v)}&\leq \int_{2\alpha'}^1 f_i(t)dt = \opt_i - \int_0^{2\alpha'}f_i(t)dt \\
    &= \opt_i - \left(\alpha'+2\alpha'^2\left(\frac{1}{2}-\xi_i\right)\right) < \opt_i -\alpha',
    \end{align*}
    where the last inequality follows by the upper bound on $\alpha'$ and $\delta$, noticing that $\xi_i\leq \delta$. We get a contradiction with $M$ being $\alpha'$ optimal on $\cD_i$, and the lower bound on the measure of $B_i$ holds.
    
    Now set $\rho_1 = \nicefrac{\xi_j}{\xi_i} = \delta^{j-i}\leq \delta$. Notice that the interval $[\rho_1, \nicefrac{1}{3}]$ has measure at least $\nicefrac{1}{3}-\delta > \nicefrac{1}{3}-\nicefrac{1}{32} > \nicefrac{1}{16}>2\alpha'$, which implies $B_i\cap [\rho_1, \nicefrac{1}{3}]\neq \emptyset$, since $B_i$ has measure at least $1-2\alpha'$. Pick $t_0\in B_i\cap [\rho_1, \nicefrac{1}{3}]$; by definition of $B_i$, we have $v_i(t_0) = \left(t_0\xi_i, \nicefrac{1+t_0}{2}\right)\in A_M$. Now, since $t_0\geq \rho_1$, we have $t_0\xi_i\geq \rho_1\xi_i = \xi_j$. Considering any $s\in [0, 1]$, we get $s\xi_j\leq t_0\xi_i$, which by monotonicity of $A_M$ implies $(s\xi_j, \nicefrac{1+t_0}{2})\in A_M$. By definition of Myerson payments, we thus get the buyer payment satisfying $q_M(s\xi_j, z)\leq \nicefrac{1+t_0}{2} \leq \nicefrac{1}{2}+\nicefrac{1}{6} = \nicefrac{2}{3}$ for any buyer coordinate $z \in [0, 1]$. This implies that $\prof(M, v_j(s))\leq q_M(v_j(s))\leq \nicefrac{2}{3}$. Taking the expectation over $\cD_j$, \Cref{eq:prof_up_lower_bound} holds.

    Assume now that $i>j$, so that $\ell_i$ is now the steeper line of the two. Let $\rho_2 = \frac{\xi_i}{\xi_j} = \delta^{i-j}\leq \delta$. We first prove that $A_M\cap \{v_j(s), s\in [\rho_2, \nicefrac{1}{2}]\} = \emptyset$. We prove this by contradiction. Assume that $v_j(s) = (s\xi_j, \nicefrac{1+s}{2})\in A_M$ for some $s\in [\rho_2, \nicefrac{1}{2}]$. Since $s\geq \rho_2$, we have $s\xi_j\geq \rho_2\xi_j = \xi_i$, which in turn implies, for any $t\in [0, 1]$, $t\xi_i\leq \xi_i\leq s\xi_j$. As for the first case of this proof, we then use monotonicity of $A_M$ to claim that $(t\xi_i, \nicefrac{1+s}{2})\in A_M$, for any $t\in [0, 1]$. Consequently, $q_M(t\xi_i, z)\leq \nicefrac{1+s}{2}$ for any $z\in [0, 1]$. Again by monotonicity, $v_i(t) = (t\xi_i, \nicefrac{1+t}{2}) \in A_M$ for any $t\geq s$, meaning that $p_M(t\xi_i, \nicefrac{1+t}{2})\geq t\xi_i$. Together, the upper bound on $q_M$ and the lower bound on $p_M$ imply that $\prof(M, v_i(t))\leq \nicefrac{1+s}{2}-t\xi_i$, for any $t\in [s, 1]$. The average loss with respect to the maximum profit achievable for $\cD_i$ is then lower bounded by: 
    $$
    \int_{s}^1f_i(t)-\prof(M, v_i(t))dt\geq \int_s^1\frac{t-s}{2}dt = \frac{(1-s)^2}{4} \geq \frac{1}{16}> \alpha',
    $$
    where we have used that $s\leq \nicefrac{1}{2}$ and that $\alpha'<\nicefrac{1}{32}$. This contradicts $\alpha'$-optimality of $M$ on $\cD_i$, so $A_M\cap \{v_j(s), s\in [\rho_2, \nicefrac{1}{2}]\}$ is empty. Rephrasing, $A_M\cap \ell_j \subseteq \{v_j(s), s\in [0, \rho_2)\cup(\nicefrac{1}{2}, 1]\}$, but the set $ C= [0, \rho_2)\cup(\nicefrac{1}{2}, 1]$ has measure $\nicefrac{1}{2}+\rho_2\leq \delta+ \nicefrac{1}{2}<\nicefrac{2}{3}$, where we are using that $\delta<\nicefrac{1}{32}$. It follows that \Cref{eq:prof_up_lower_bound} must hold, since, using that $\prof$ is upper bounded by $1$, we have $\Esub{v\sim \cD_j}{\prof(M, v)}= \int_C \prof(M, v_j(s))ds\leq \nicefrac{2}{3}$.
\end{proof}

\subsection{Missing Proofs from \Cref{subsec:restriction_profit}}
\label{app:missing_net}

    \cardbound*
        \begin{proof}
            Consider the directed graph $G_{\eta}$ whose nodes are given by the uniform grid $V_{\eta}$, and whose edges are the sides of the tiles of $\cT_{\eta}$, pointing either from left to right (for horizontal edges) or bottom-up (for vertical ones). 
            Each $\eta$-simple mechanism is identified by a path on this graph going from $(0,0)$ to $(1,1)$, which --- together with the left and upper side of $[0,1]^2$ --- encloses the allocation region. We then only need to count the number of paths of this type. Each path contains at most $\nicefrac 2 \eta$ edges, equally divided into horizontal and vertical edges. All in all, we have: 
            \[
                |\M_{\eta}| = \binom{\nicefrac 2 \eta}{\nicefrac 1 \eta} \le \left(2e\right)^{\nicefrac{1}{\eta}},    
            \]
            where we use the inequality $\binom{n}{k} \le \left(\nicefrac{en}{k}\right)^{k}$. 
        \end{proof}
        
    We start by defining the boundary of an allocation region. We define it in the following way, for a monotone region $A$: 
    $$
    \partial A = \{(x, y)\in A: x = \max \{z\in [0, 1]: (z, y)\in A\}\lor y = \min \{z\in [0, 1]: (x, z)\in A\}\}
    $$
    Intuitively, the boundary is the set of points whose corresponding coordinate is the Myerson payment for a specific point in $A$. We also define $K_{M, \eta}$ as the subset of $\eta$-grid tiles in $\cT_\eta$ which have a non-empty intersection with $\partial A_M$. To avoid double counting and ambiguity, if the intersection with the tile $B$ is only a vertical or horizontal segment, and the segment is shared with the boundary of another tile $B'$, then $B\in K_{M, \eta}$ if it is to the left or to the top of $B'$, otherwise we let $B'\in K_{M, \eta}$.

    \begin{claim}
        \label{cl:cardinality}
            Consider any mechanism $M \in \M$, then $\P{v\in A_{\phi_\eta(M)}\setminus A_M} \le \nicefrac{2\eta}{\sigma}.$
        \end{claim}

\begin{proof}
    The key point is proving that if $(x, y), (x', y')\in \partial A_M$ and $x'\geq x$, then $y' \geq y$. By contradiction, assume that $y'<y$. By definition of monotone region, it must then be that $(x', y), (x, y')\in A_M$. If this is the case, then $x \neq \max \{z\in [0, 1]: (z, y)\in A\}$ and $y\neq \min \{z\in [0, 1]: (x, z)\in A\}$. Given that none of these two conditions hold, $(x, y)\notin \partial A_M$ and we reach a contradiction. $\partial A_M$ has thus only $\nicefrac{1}{\eta}$ \textit{right} movements and $\nicefrac{1}{\eta}$ \textit{upward} movements of variation $\eta$ available, otherwise we would be contradicting either $\partial A_M \subset [0, 1]^2$ or the monotone property. Moreover, each of these movements can count only once in increasing the cardinality of $K_{M, \eta}$. There are therefore only at most $\nicefrac{2}{\eta}$ cells in $K_{M, \eta}$, and each square in $K_{M, \eta}$ is induced by the $\eta$-grid, so it has Lebesgue measure $\eta^2$. Since $ A_M \setminus A_{\phi_\eta(M)}\subseteq \cup_{B\in K_{M, \eta}} B$, we
        We thus get: 
         $$
             \P{v \in A_{\phi_\eta(M)}\setminus A_M} \leq \P{v \in \bigcup_{B\in K_{M, \eta}}B} \leq\frac{1}{\sigma} \cdot \cL\left(\bigcup_{B\in K_{M, \eta}}B\right) \leq \frac{1}{\sigma} \cdot \eta^2 \frac{2}{\eta} = \frac{2\eta}{\sigma}
         $$
        
\end{proof}

We now proceed to proving that the number of columns and rows in the $\eta$-grid resulting in a significant approximation error is contained and can thus be neglected, allowing the proof of \Cref{lem:lp_net}.

    \begin{claim}
    \label{cl:bound_cl_ro}
        Given that $\eta = 2^{-H}$ for some integer $H$, we have that $\max(\lvert C_i\rvert, \lvert R_i\rvert)\leq 2^{i+2}$.
    \end{claim}
    \begin{proof} 
    The only relevant case is $i\leq H-3$, otherwise the statement is trivial and follows since there are at most $\nicefrac 1 \eta$ columns or rows.
     We now argue that $\lvert C_i\rvert\leq 2^{i+2}$, but an identical reasoning shows that the same holds for $R_i$. 
    Notice that, if $M_c \in [2^{-i-1}, 2^{-i})$ for $c\in C$, $\partial A_M$ must have accumulated at least $2^{-i-1}-\eta$ vertical variation inside $c$. If this were not the case, there would be an upper tile from $\cT_\eta$ in $A_M$ inducing a smaller error and thus a contradiction. Since $\partial A_M$ is monotone, it has variation bounded by $1$, resulting in $1\geq \lvert C_i\rvert\cdot (2^{-i-1} -\eta)$. Manipulating the inequality, we get $\lvert C_i\rvert\leq \frac{2^{i+1}}{1-\eta\cdot 2^{i+1}}$. Since $i< H-2$ and $\eta = 2^{-H}$, then the denominator is bounded below by $\nicefrac{1}{2}$ and we get the result.
    \end{proof}

\subsection{Missing Proofs from \Cref{subsec:unif_convergence}}
\label{app:missing_uniform}

\telescope*

\begin{proof}
        
    \noindent \textbf{Step 1: Symmetrization} 
    Let $S'$ be an i.i.d. sample from the same distribution as $S$. By linearity, we have the equality $\Esub{v}{\prof(M, v)} = \Esub{S'}{\prof(M, S')}$. Therefore, we have the following: 

    \begin{align}
        \Esub{S}{\max_{M \in \hat\M} \Bigl\lvert\prof(M,S) - \Esub{v}{\prof(M, v)}\Bigr\rvert} & = \Esub{S}{\max_{M \in \hat\M}\Bigl\lvert\Esub{S'}{\prof(M,S) - \prof(M, S')}\Bigr\rvert} \notag\\
        & \leq \Esub{S, S'}{\max_{M \in \hat\M}\Bigl\lvert\prof(M,S) - \prof(M, S')\Bigr\rvert} \label{eq:step1telescope}
    \end{align}
    In the inequality we employed Jensen’s inequality twice: the absolute value function is convex, and the maximum function is also convex, as a map $\mathbb{R}^{\lvert \hat\M\rvert} \rightarrow \mathbb{R}$.
    
    \medskip
    \noindent \textbf{Step 2: Rademacher Complexity} Now, let $\boldsymbol{r} = \{r_i\}_{i=1}^n$ be a vector of independent Rademacher random variables\footnote{The Rademacher distribution puts equal probability on $-1$ and $1$.}. The following inequality holds.

    \begin{equation}
    \label{eq:step2telescope}
    \Esub{S, S'}{\max_{M \in \hat \M}\Bigl\lvert\prof(M,S) - \prof(M, S')\Bigr\rvert} \leq 2\Esub{S, \boldsymbol{r}}{\max_{M \in \hat \M}\left\lvert \frac{1}{n}\sum_{i=1}^n\prof(M,v_i)\cdot r_i\right\rvert}
    \end{equation}

    The proof follows from an equality in distribution. Specifically, we can argue that the following holds: 
    \begin{equation*}
        \{\prof(M,v_i) - \prof(M, v'_i)\}_{M\in \hat\M,\, i\in [n]} \overset{(d)}{=} \{(\prof(M,v_i) - \prof(M, v'_i))\cdot r_i\}_{M\in \hat\M,\, i\in [n]}
    \end{equation*}

   Proving the equality above boils down to the following observation: swapping one variable in $S$ with one variable in $S'$ is irrelevant. Let $T = S \cup S'$. At each iteration $i$, we select one pair of variables at a time from the set $T$, inducing another partitioning, made up of the sets $T_1 = \{t^1_i\}_{i=1}^n$ and $T_2 = \{t^2_i\}_{i=1}^n$. If $r_i = -1$, then we let $t^1_i = v'_i$ and $t^2_i = v_i$; if $r_i=1$ the opposite happens. 
    
    Now, we can define a function $g: [0,1]^{2n} \rightarrow [-2, 2]^{n\cdot \lvert \hat\M \rvert}$ such that $g(T_1, T_2) = \{\prof(M,t_i^1) - \prof(M, t^2_i)\}_{M\in \hat\M,\, i\in [n]}$. For any Borel set $B \subseteq [-2, 2]^{n\cdot \lvert \hat\M \rvert}$, we have:
    $$
    \P{g(S, S') \in B} = \P{g(T_1, T_2) \in B\, | \, \boldsymbol{r}},
    $$ 
    since $T$ is i.i.d.. We thus get $\P{g(S, S') \in B}= \P{g(T_1, T_2) \in B}$.
    The claimed equality in distribution then follows by noticing that $\{\prof(M,t_i^1) - \prof(M, t^2_i)\} = \{(\prof(M,v_i) - \prof(M, v'_i))\cdot r_i\}$, by the way we have designed the splitting procedure of $T$ in $T_1$ and $T_2$.

    To conclude, we compose together the maximum function over $\hat\M$, the absolute value function and the sample averaging, and then apply the result onto the two equally distributed random vectors, resulting in the following, which suffices to complete the proof for this step: 
    $$
    \max_{M\in \hat\M}\lvert \prof(M, S)-\prof(M, S')\rvert \overset{(d)}{=} \max_{M\in \hat\M}\left\lvert \frac{1}{n}\sum_{i=1}^n(\prof(M, v_i)-\prof(M, v'_i))\cdot r_i\right\rvert
    $$

    \medskip
    \noindent \textbf{Step 3: Gaussian Complexity} The next step is relating the Rademacher complexity to the Gaussian complexity. Specifically, let $\boldsymbol{g} = \{g_i\}_{i=1}^n$ be a vector of standard Gaussian variables, the following inequality then holds:

    \begin{equation}
    \label{eq:step3telescope}
    \Esub{S, \boldsymbol{r}}{\max_{M \in \hat \M}\left\lvert \frac{1}{n}\sum_{i=1}^n\prof(M,v_i)\cdot r_i\right\rvert} \leq \sqrt{\frac{\pi}{2}}\Esub{S, \boldsymbol{g}}{\max_{M \in \hat \M}\left\lvert \frac{1}{n}\sum_{i=1}^n\prof(M,v_i)\cdot g_i\right\rvert}
    \end{equation}
    By a reasoning almost identical to what we did for Step 2, now using symmetry of the gaussian law instead of the exchangeability of the i.i.d. draw, we get the following equality in the joint distributions: 
    \begin{equation*}
        \{\prof(M,v_i)\cdot g_i\}_{M\in \hat\M,\, i\in [n]} \overset{(d)}{=} \{\prof(M,v_i)\cdot \lvert g_i\rvert \cdot r_i\}_{M\in \hat\M,\, i\in [n]}
    \end{equation*} 
    
    Now we use this equality in distribution to prove the claimed inequality.

    \begin{align*}
        & \Esub{S, \boldsymbol{g}}{\max_{M\in \hat\M}\left\lvert \frac{1}{n}\sum_{i=1}^n(\prof(M, v_i)-\prof(M, v'_i))\cdot g_i \right\rvert}\\ & = \Esub{S, \boldsymbol{g}, \boldsymbol{r}}{\max_{M\in \hat\M}\left\lvert \frac{1}{n}\sum_{i=1}^n(\prof(M, v_i)-\prof(M, v'_i))\cdot \lvert g_i \rvert \cdot r_i\right\rvert} \\
        & \ge \Esub{S, \boldsymbol{r}}{\max_{M\in \hat\M}\left\lvert \frac{1}{n}\sum_{i=1}^n(\prof(M, v_i)-\prof(M, v'_i))\cdot \E{\lvert g_i\rvert}\cdot r_i\right\rvert} \\
        & = \sqrt{\frac{2}{\pi}}\Esub{S, \boldsymbol{r}}{\max_{M\in \hat\M}\left\lvert \frac{1}{n}\sum_{i=1}^n(\prof(M, v_i)-\prof(M, v'_i))\cdot r_i\right\rvert} \tag{$\E{\lvert g_i\rvert} = \sqrt{\nicefrac{2}{\pi}}$},
    \end{align*}
    where in the second to last inequality we have used Jensen's inequality.

    \medskip
    \noindent \textbf{Step 4: Telescoping} We complete the proof by telescoping $\prof$ for a mechanism in $\hat\M$ across all the levels of approximations (i.e. with $h$ ranging from $0$ to $H = \log_2(\nicefrac{8}{\acc\sigma})$). Our target is the following inequality.

    \begin{align}
    & \Esub{S, \boldsymbol{g}}{\max_{M\in \hat\M}\left\lvert \frac{1}{n}\sum_{i=1}^n(\prof(M, v_i)\cdot g_i\right\rvert} - \frac{2}{\sqrt{\pi n}}\notag\\
    &\leq \sum_{h=0}^{H-1}\Esub{S,\textbf{g}}{\max_{M \in \hat \M} \left\lvert \frac{1}{n}\sum_{i=1}^n(\prof(\phi_{h+1}(M),v_i) - \prof(\phi_h(M),v_i)) \cdot g_i\right\rvert}\label{eq:step4telescope}
    \end{align}

    As we were hinting at before, for any mechanism in $\hat\M$, we have that $\phi_H(M)=M$. Therefore, for any point $v\in [0, 1]^2$, we have: 
    \begin{equation*}
        \prof(M, v) = \prof(\phi_0(M), v)+\sum_{h=0}^{H-1}(\prof(\phi_{h+1}(M), v)-\prof(\phi_h(M), v))
    \end{equation*}
    By \Cref{def:approximations_prof}, we have that $\phi_0(M)$ is a constant function, either equal to $-1$ or $0$. Therefore, the previous decomposition suffices to show, by the triangle inequality on the maximum function and the absolute value function, this inequality:
    \begin{align*}
        &\max_{M\in \hat\M}\left\lvert \frac{1}{n}\sum_{i=1}^n(\prof(M, v_i))\cdot g_i\right\rvert \\
        & = \max_{M\in \hat\M}\left\lvert \frac{1}{n}\sum_{i=1}^n\prof (\phi_0(M), v_i)\cdot g_i + \frac{1}{n}\sum_{i=1}^n\sum_{h=0}^{H-1}(\prof(\phi_{h+1}(M), v_i)-\prof(\phi_h(M), v_i))\cdot g_i\right\rvert \\
        &\leq \max_{M\in \hat\M}\left\lvert \frac{1}{n}\sum_{i=1}^n\prof (\phi_0(M), v_i)\cdot g_i\right\rvert + \max_{M\in \hat\M} \left\lvert \frac{1}{n}\sum_{i=1}^n\sum_{h=0}^{H-1}(\prof(\phi_{h+1}(M), v_i)-\prof(\phi_h(M), v_i))\cdot g_i\right\rvert\\
        & \leq \frac{1}{n}\left \lvert \sum_{i=1}^n g_i\right \rvert +\sum_{h=0}^{H-1}\max_{M \in \hat \M} \left\lvert \frac{1}{n}\sum_{i=1}^n(\prof(\phi_{h+1}(M),v_i) - \prof(\phi_h(M),v_i)) \cdot g_i\right\rvert
    \end{align*}
    To bound the expectation of the first term, we take the expectation of the absolute value of a gaussian variable with variance equal to $n$, which is equal to $\sqrt{\nicefrac{2n}{\pi}}$.

    \medskip
    \noindent \textbf{Step 4: Wrapping up} Combining \Cref{eq:step1telescope,eq:step2telescope,eq:step3telescope,eq:step4telescope}, we immediately obtain the claim.
\end{proof}

\unifconv*
    \begin{proof}
    By using the law of total expectation, we have the following inequality.
        \begin{align}
        \nonumber
        &\Esub{S}{\max_{M \in \Mnet} \Bigl\lvert\prof(M,S) - \Esub{v}{\prof(M, v)}\Bigr\rvert} - \frac{2}{\sqrt n}\\
        \nonumber
        & \leq \sqrt{2\pi}\sum_{h=0}^{H-1}\Esub{S,\textbf{g}}{\max_{M \in \Mnet} \left\lvert \frac{1}{n}\sum_{i=1}^n(\prof(\phi_{h+1}(M),v_i) - \prof(\phi_h(M),v_i)) \cdot g_i\right\rvert} \tag{by \Cref{cl:unif_to_telescope}} \\
        \nonumber
        &\le \sqrt{2\pi}\sum_{h=0}^{H-1}\Esub{S,\textbf{g}}{\max_{M \in \Mnet} \left\lvert \frac{1}{n}\sum_{i=1}^n(\prof(\phi_{h+1}(M),v_i) - \prof(\phi_h(M),v_i)) \cdot g_i\right\rvert \bigg | \clean} \\
        & \quad +\sqrt{2\pi}\sum_{h=0}^{H-1}\Esub{S,\textbf{g}}{\max_{M \in \Mnet} \left\lvert \frac{1}{n}\sum_{i=1}^n(\prof(\phi_{h+1}(M),v_i) - \prof(\phi_h(M),v_i)) \cdot g_i\right\rvert \bigg | \clean^c} \P{\clean^c}
        \label{ineq:step_0_final}
        \end{align}
        We bound the two terms in the right-most term separately. For the term corresponding to the clean event, we simply apply the bound of \Cref{cl:telescop_conc} for each of the $H=\log_2({\nicefrac8{\sigma\alpha}})$ levels: 
        \begin{equation}
            \label{eq:step_1_final}
            \sqrt{2\pi}\sum_{h=0}^{H-1}\Esub{S,\textbf{g}}{\max_{M \in \Mnet} \left\lvert \frac{1}{n}\sum_{i=1}^n(\prof(\phi_{h+1}(M),v_i) - \prof(\phi_h(M),v_i)) \cdot g_i\right\rvert \bigg | \clean} \leq 22\sqrt{\frac{2\pi}{{\sigma n}}}\cdot \log_2\left({\frac8{\sigma\alpha}}\right)
        \end{equation}

        To take care of the second term, corresponding to the complement of the clean event, we have a coarser bound on the variance of the gaussian variables involved, but a stringent bound on $\P{\clean^C}$ (by \Cref{cl:clean_event_control}). First, fix any realization of the sample $S$, the random variables
        \[
            \frac{1}{n}\sum_{i=1}^n(\prof(\phi_{h+1}(M),v_i) - \prof(\phi_h(M),v_i)) \cdot g_i 
        \]
        are still gaussian variables, but we can only apply a coarser bound on the variance:
        \begin{align*}
            &\Var{\frac{1}{n}\sum_{i=1}^n(\prof(\phi_{h+1}(M),v_i) - \prof(\phi_h(M),v_i)) \cdot g_i} \\
            &\quad \le \frac{1}{n}\max_{i=1}^n (\prof(\phi_{h+1}(M),v_i) - \prof(\phi_h(M),v_i))^2\le \frac{4}{n}.
        \end{align*}
        Concerning the cardinality of $\Mnet$, we have the upper bound provided by \Cref{prop:card_bound}: $|\Mnet| \le (2e)^{\nicefrac{8}{\sigma \acc}}$. We can therefore use once again the folklore bound of the $\max$ of the absolute value of gaussian variables (\Cref{prop:max}) to get:
        \begin{align}
            \sqrt{2\pi}&\sum_{h=0}^{H-1}\Esub{S,\textbf{g}}{\max_{M \in \Mnet} \left\lvert \frac{1}{n}\sum_{i=1}^n(\prof(\phi_{h+1}(M),v_i) - \prof(\phi_h(M),v_i)) \cdot g_i\right\rvert \bigg | \clean^c} \P{\clean^c}
            \nonumber \\
            &\le \sqrt{2\pi} \cdot 16 \cdot \log_2({\nicefrac8{\sigma\alpha}}) \cdot \frac{1}{\sqrt{n\sigma \acc}}\P{\clean^c} \le \sqrt{2\pi}\cdot 16 \cdot  \log_2({\nicefrac8{\sigma\alpha}})\cdot \sqrt{\frac{\alpha}{\sigma n}},
            \label{eq:step_2_final}
        \end{align}
        where the last inequality follows from \Cref{cl:clean_event_control}. Plugging in \Cref{eq:step_1_final,eq:step_2_final} into \Cref{ineq:step_0_final}, we get:
        \begin{equation*}
            \Esub{S, v}{\max_{M \in \Mnet} \Bigl\lvert\prof(M,S) - \E{\prof(M, v)}\Bigr\rvert} \leq \sqrt{2\pi}\log_2({\nicefrac8{\sigma\alpha}})\left(16\sqrt{\frac{\alpha}{\sigma n}} + \frac{24}{\sqrt{\sigma n}}\right) \le 40 \log_2({\nicefrac8{\sigma\alpha}})\sqrt{\frac{2 \pi}{\sigma n}}.
        \end{equation*}
        The statement of the claim follows by recalling the lower bound on the sample number $n\ge \frac{35^3}{\sigma\alpha^2}$, which implies $\alpha \ge 210\sqrt{\nicefrac{1}{\sigma n}}$.
    \end{proof}

\subsection{Missing Proofs from Section 3.4}
\label{subsec:missing35}
The primary feature of the set of edge weights defined in \Cref{def:weight_edges} is described by the following lemma, which establishes a relationship between the probability of a complete path $\pi$ under the measure $\mu$ and the weight of each edge that constitutes it. 
    \begin{lemma}
    \label{lem:decompose_path}
    Let $S=\{v_i\}_{i=1}^n$ be a sample of valuations from a smooth distribution. For every path $\pi \in \pathspace$, let $w_\pi = \exp(n\cdot \nicefrac{\eps}{4}\cdot \prof(M_\pi, S))$. Then the following equality almost surely holds: $w_\pi = \prod_{e\in \pi}w_e$.
    \end{lemma}
    \begin{proof}[Proof of \Cref{lem:decompose_path}]
        Proving the lemma is equivalent to proving that the following equality almost surely holds:
        \begin{equation}
        \label{eq:prof_decomposition_graph}
        \sum_{i=1}^n\prof(M_\pi, v_i) = \sum_{e\in \pi} c_e\lvert A_e \cap S\rvert,
        \end{equation}
        where $c_e$ is $-{k}{\eta}$ if $e$ is a vertical edge and ${j}{\eta}$ if $e$ is a horizontal edge. Let $\pi_h$ be the horizontal edges in $\pi$ and $\pi_v$ the vertical ones in $\pi$. To prove \Cref{eq:prof_decomposition_graph}, we first observe that the sets $\{A_e\}_{e\in \pi_v}$ and $\{A_e\}_{e \in \pi_h}$ form a partition of the allocation region of $M_\pi$, provided we disregard a finite number of null probability sets (the boundary of $[0, 1]^2$ and segments shared by pairs of interest regions).
        
        Consider a point $v \in S$. If $v \notin A_{M_\pi}$, then $M_\pi$ does not extract any profit from it and $v\notin A_e\, \forall\, e\in\pi$, consistently with \Cref{eq:prof_decomposition_graph}. If $v\in A_{M_\pi}$, then $v\in A_e\cap A_{e'}$, for a specific horizontal $e$ and a specific vertical $e'$ in $\pi$, and $\prof$ is exactly $c_e + c_{e'}$, by \Cref{def:payments} for Myerson payments.
        \end{proof}

    As intuitively explained in \Cref{subsec:wrapping}, the node weights from \Cref{def:weight_nodes} satisfy the property that $w_u$ is exactly the sum of the weights of all the paths that stem from $u\in V_\eta$, as proven in the following lemma.
    \begin{lemma}
    \label{lem:sum_future}
        For any node $u\neq (1, 1)$ in $V_\eta$, let $\pathspace^u$ be the set of paths from $u$ to $(1, 1)$. Then the following holds: $w_u = \sum_{\pi \in \pathspace^u}\prod_{e\in \pi}w_e$.
    \end{lemma}
    \begin{proof}
        We prove the statement by induction over the unweighted shortest path distance $d$ from $(1, 1)$. If $d=1$, then there are only two nodes to consider, the ones directly connected to $(1, 1)$, meaning $(1-\eta, 1)$ and $(1, 1-\eta)$. The base case of the induction thus holds by \Cref{def:weight_nodes}, since for each node there is only one relevant edge and $w_{(1, 1)} =1$.
    
        For the inductive step, assume that the statement is true for all $d$ smaller than some $d'$, and consider a vertex $u$ at distance $d'$ from $(1, 1)$. Then the following holds: 
        $$
        w_u = \!\!\sum_{s\in \neigh{u}}w_{(u, s)}w_s = \!\sum_{s\in \neigh{u}}\!\!w_{(u, s)} \sum_{\pi\in \pathspace^s}\prod_{e\in \pi}w_e = \sum_{\pi \in \pathspace ^u}\prod_{e\in \pi}w_e, 
        $$
        where we have used the definition of $w_u$ in the first equality and the inductive hypothesis in the second, since $s$ must have distance $d= d'-1$ from $(1, 1)$ due to the way we have built the graph. The third follows by rearranging and noticing that: 
        $$
        \bigcup_{s\in \neigh{u}}\bigcup_{\pi \in \pathspace^s}\{\pi \cup (u, s)\} = \pathspace^u.
        $$ 
    \end{proof}
    
    We now have the two crucial ingredients to prove \Cref{theor:sampling}, we restate it here for convenience. 
    \sampling*
    \begin{proof}
    Point (i) and (ii) claimed in \Cref{theor:sampling} are immediate from the definition of $q_e$ in \Cref{eq:prob_definition}. Additionally, as stressed in \Cref{subsec:wrapping}, computing $\{q_e\}_{e\in E_\eta}$ takes only $O(\nicefrac{n}{\eta^2})$ time. The computation of the edge weights from \Cref{def:weight_edges} can be implemented by jointly scanning $S$ and the uniform $\eta$-grid, in $O(\nicefrac{n}{\eta^2})$ time; the node weights from \Cref{def:weight_nodes} can then be computed by dynamic programming, in $O(\nicefrac{1}{\eta^2})$ time.

    It remains to prove point (iii). By the chain rule of conditional probabilities, this would follow from the following equality, assuming that $(e_1, \dots, e_{i-1})$ is a path and that $\Pi \sim \mu$: 
    $$
    \mu((e_1, e_2, \dots, e_i)\in \Pi\, |\, (e_1, e_2, \dots, e_{i-1})\in \Pi) = q_{e_i}.
    $$
    We dedicate the remainder of the proof to proving the above. If $e_i$ is not connected to the endpoint of $e_{i-1}$, then both sides are zero. If $e_i$ is the only possible edge after $e_{i-1}$, meaning that the endpoint of $e_{i-1}$ has only one outgoing edge, then the equality is satisfied by definition of $q_{e_i}$: both sides of the equality are $1$. Therefore, let $u$ be the endpoint of $e_{i-1}$ and let $\neigh{u} = \{r, s\}$. Consider $e_i = (u, r)$, but the argument below holds the same for $(u, s)$.

    Let $\pathspace^{:i}\subseteq \pathspace$ be the set of complete paths that have $(e_1, e_2, \dots, e_i)$ as their prefix. Recall that $w_\pi = \exp(n\cdot \nicefrac{\eps}{4}\cdot \prof(M_\pi, S))$. Then, by definition of conditional probability:
    \begin{align*}
        \mu((e_1, e_2, \dots, e_i)\in \Pi\, |\, (e_1, e_2, \dots, e_{i-1})\in \Pi) &= \frac{\mu((e_1, e_2, \dots, e_i)\in \Pi)}{\mu((e_1, e_2, \dots, e_{i-1})\in \Pi)} \\
        & = \frac{\sum_{\pi \in \pathspace^{:i}}w_\pi}{\sum_{\pi' \in \pathspace^{:i-1}}w_{\pi'}}, 
    \end{align*}
    where the last equality follows by the way we have defined $\mu$ in the statement of \Cref{theor:sampling}. By the same definition, the denominator is always non-zero.
    
    To conclude, we use \Cref{lem:decompose_path,lem:sum_future}. Let $\pi_{i:}$ be the suffix of the complete path $\pi$ from the $i$-th edge included, and let $\pathspace^{: \overline i}\subseteq \pathspace$ be the set of complete paths having $(e_1, e_2, \dots, (u, s))$ as their prefix. We have this almost sure (with respect to the sampling of $S$) chain of equalities: 
    \begin{align*}
        \frac{\sum_{\pi \in \pathspace^{:i}}w_\pi}{\sum_{\pi' \in \pathspace^{:i-1}}w_{\pi'}} &= \frac{\sum_{\pi \in \pathspace^{:i}}\prod_{e\in \pi}w_e}{\sum_{\pi' \in \pathspace^{:i-1}}\prod_{e'\in \pi'}w_{e'}} \tag{\Cref{lem:decompose_path}} \\
        &= \frac{\prod_{j=1}^{i} w_{e_j}\sum_{\pi \in \pathspace^{:i}}\prod_{e\in \pi_{i+1:}}w_e}{\prod_{j'=1}^{i-1} w_{e_{j'}}\sum_{\pi' \in \pathspace^{:i-1}}\prod_{e'\in \pi'_{i:}}w_{e'}} \\
        &=\frac{w_{(u,r)}\sum_{\pi\in \pathspace^{:i}}\prod_{e\in \pi_{i+1:}}w_e}{w_{(u, s)}\sum_{\pi'\in \pathspace^{:\overline i}}\prod_{e'\in \pi'_{i+1:}}w_{e'}+w_{(u, r)}\sum_{\pi''\in \pathspace^{:i}}\prod_{e''\in \pi''_{i+1:}}w_{e''}} \\
        &=\frac{w_{(u, r)}w_r}{w_{(u, s)}w_s + w_{(u, r)}w_r} = q_{e_i}. \tag{\Cref{lem:sum_future}}
    \end{align*}

    \end{proof}

    \subsection[Missing Proofs from Section 4.1]{Missing Proof from \Cref{subsec:lower_gft}}
    \label{app:lower_gft}
    
    \gftimpossible*
        \begin{proof}
            We begin by constructing a ``hard'' family of distributions, indexed by the uniform grid $\Gamma$ over the interval $[\nicefrac{9}{20},\nicefrac{11}{20}]$, with step-size $\xi$.
            For simplicity, we assume that $\nicefrac 1{(20\xi)}$ is an integer; we set the precise value of $\xi$ later. The generic $\gamma_i$ in $\Gamma$, is then of the form $\nicefrac{9}{20}+ i \xi$.

            For each $\gamma_i \in \Gamma$, we introduce an associated distribution $\cD_i$. Sampling a valuation $v$ from $\cD_i$ works in two phases: first, an auxiliary random variable $x$ is sampled, then the valuation $(0,x)$ is drawn if $x\ge \gamma_i$, and the valuation $(x,1)$ is drawn otherwise. The auxiliary variable $x$ is equally likely to be drawn uniformly at random from $[\nicefrac{9}{20},\nicefrac{11}{20}]$, or from $I_i=[\gamma_i- \nicefrac{\xi}{3},\gamma_i + \nicefrac{\xi}{3}]$. The optimal price given $\cD_i$ is $\gamma_i$ by design, for an expected gain from trade of at least $\nicefrac{9}{20}$. Indeed, by construction, $\gamma_i$ always makes the trade happen, and the difference between the buyer and the seller valuation is at least $\nicefrac{9}{20}:$
            \begin{equation}
            \label{eq:opt_GFT}
                \sup_{p \in [0,1]} \Esub{v \sim \cD_i}{\gft(p,v)} = \Esub{v \sim \cD_i}{\gft(\gamma_i,v)} \ge \frac{9}{20}.
            \end{equation}
                
            Assume for the sake of contradiction that there exists a learning algorithm $\A$ that is $\eps$-differentially private, and that guarantees expected gain from trade at most $\alpha \le \nicefrac{1}{80}$ far from the optimal one, for any $\cD_i$. Then we should have, by exploiting the lower bound in \Cref{eq:opt_GFT}:
            \begin{equation}
            \label{eq:alg_gft}
                \Esub{v \sim \cD_i}{\gft(\A(S),v)} \ge \sup_{p \in [0,1]} \Esub{v \sim \cD_i}{\gft(p,v)} - \alpha \ge \frac{35}{80},
            \end{equation}
            where the expectation is taken with respect to both $S$, $v$ (from $\cD_i$), and the internal randomness of $\A$. We now prove a technical claim we use in the rest of the proof. 
            \begin{claim}
            \label{cl:S_i}
                Consider any distribution $\cD_i$. If $\A$ respects \Cref{eq:alg_gft}, then there exists at least one realization of the sample $S$ such that $\A(S) \in I_i$ with probability {at least $\nicefrac{2}{11}$}.
            \end{claim}
            \begin{proof}[Proof of \Cref{cl:S_i}]
                Conditioning on $\A(S) \notin I_i$, we have:
                \begin{equation}
                \label{eq:gft_bad_event}
                    \Esub{v \sim \cD_i}{\gft(\A(S),v)|\A(S) \notin I_i} \le \frac{33}{80}.
                \end{equation}
                Indeed, if the posted price is not within the right interval $I_i$, then it misses at least half of the trades when $x$ is drawn from $I_i$. This happens with probability at least one quarter, while for every trade that actually takes place, the gain from  trade is at most $\nicefrac{11}{20}.$ We have the following chain of inequalities:
                \begin{align*}
                \frac{35}{80} &\le \E{\gft(\A(S),v)} \tag{By \Cref{eq:alg_gft}}\\
                &= \E{\gft(\A(S),v)|\A(S)\in I_i} \P{\A(S) \in I_i} + \E{\gft(\A(S),v)|\A(S)\notin I_i} \P{\A(S) \notin I_i}\\
                &\le \frac{{11}}{20} (1-\P{\A(S) \notin I_i}) + \frac{33}{80}\P{\A(S) \notin I_i}. \tag{By \Cref{eq:gft_bad_event}}
                \end{align*}
                By rearranging the above inequality, we can conclude that ${\P{\A(S) \in I_i} \ge \nicefrac{2}{11}}.$
            \end{proof}

            We want to argue that the result obtained in the Claim leads to a contradiction. In particular, for each one of the distributions $\cD_i$ indexed by $\Gamma$, denote with $S_i$ the corresponding sample whose existence is guaranteed in the Claim. Since the intervals $I_i$ are disjoint by construction, we have
            \begin{equation}
            \label{eq:non-overlapping-I_i}
                \Psub{\A}{\A(S_i) \notin I_i} \ge \sum_{j\neq i} \Psub{\A}{\A(S_i) \in I_j}.
            \end{equation}
            We observe that the generic samples $S_i$ and $S_j$ may be different, but \emph{for sure}, they differ by at most $N$ points. Therefore, we can construct a sequence of sets $S_i=S^0_{i,j}, S^1_{i,j}, \dots, S^N_{i,j}=S_j$ such that $S^k_{i,j}$ and $S^{k+1}_{i,j}$ differ by at most one point. We can then repeatedly apply the definition of $\eps$-differentially private algorithm on $\A$:
            \begin{align}
                \Psub{\A}{\A(S_i) \in I_j} &\ge e^{-\eps} \Psub{\A}{\A(S^1_{i,j}) \in I_j}\tag{$S_i$ and $S^1_{i,j}$ differ in one point}\\
                &\ge e^{-\eps}\left[e^{-\eps} \Psub{\A}{\A(S^2_{i,j}) \in I_j}\right]  \tag{$S^1_{i,j}$ and $S^2_{i,j}$ differ in one point}\\
                &\ge e^{-N\eps}\Psub{\A}{\A(S_{j}) \in I_j}\tag{iterating $N$ times}\\
                &\ge {2}\frac{e^{-N\eps}}{{11}}. \tag{By definition of $S_i$ and \Cref{cl:S_i}} \label{eq:S_i_vs_S_j}
            \end{align}
            If we plug in the above inequality into \Cref{eq:non-overlapping-I_i}, we get the final contradiction:
            \[
                \Psub{\A}{\A(S_i) \notin I_i} \ge {2}\frac{e^{-N\eps}}{{11}} \left(\frac{1}{\xi} - 2\right) > {\frac{9}{11}},
            \]
            where the last inequality follows by choosing $\xi$ {sufficiently small}, which violates \Cref{cl:S_i}.
        \end{proof}

    \subsection[Missing Proofs from Section 4.2]{Missing Proof from \Cref{subsec:upper_gft}}
    \label{app:missing42}
    
    \uniformgft*
    \begin{proof}
    Denote with $\cE_p$ the clean event corresponding to price $p$, which depends on the realized samples $S$:
    \[
        \cE_p=\left\{\Big{|}\Esub{v}{\gft(M_p,v)} - \gft(M_p,S)\Big{|} \le \frac{\acc}4\right\}
    \]
    By the standard Hoeffding bound (\Cref{prop:chernoff}), the probability of $\cE_p^C$ is at most $2 e^{-\nicefrac{\acc^2|S|}{8}}$, so that, by union bounding over all the elements in $\hat \cG$, it holds:
    \begin{align*}
        \P{\Big(\bigcap_{p \in \hat \cG}\cE_p\Big)^C} &\le \sum_{p \in \hat \cG}\P{\cE_p^C} \le \sum_{p \in \hat \cG} 2 e^{-\nicefrac{\acc^2|S|}{8}} \le \frac{16}{\alpha\sigma}e^{-\nicefrac{\acc^2|S|}{8}} \le \frac{\beta}{2},
    \end{align*}
    where the last inequality follows by assuming that $|S| \ge \frac{8}{\alpha^2}\log\left(\frac{32}{\alpha\beta\sigma}\right)$.
\end{proof}

\section{The maximum is attained}
\label{sec:max_attained}

In this section of the appendix, we show that there exists a mechanism $M\in \mathcal M$ achieving $\prof$ equal to $\sup_{M\in \mathcal M}\E{\prof(M, v)}$, whenever we assume that the underlying distribution on $v$ is $\sigma$-smooth. To prove the result, it is convenient to switch the terminology from mechanisms to functions. Specifically, we have the following, which follows directly from \Cref{def:monotone_allocation}.

\begin{proposition}
    For any mechanism $M\in \M$, there exists a left-continuous and increasing function $f:[0, 1]\rightarrow [0, 1]$ such that $A_M$ is its epigraph\footnote{This holds if we add to the allocation region the upper boundary of $[0, 1]^2$ that is not contained in $A_M$, which has zero measure.}. Similarly, for any left-continuous and increasing function $f:[0, 1]\rightarrow [0, 1]$, there exists $M \in \M$ such that $A_M$ is its epigraph.
\end{proposition}

The previous proposition allows representing the mechanisms in $\M$ as the class of left-continuous and increasing functions, which we call $\cF$. Therefore, we may use $\prof$ taking as first argument $f\in \cF$ instead of a mechanism $M\in \M$. To build intuition, for any $M\in \M$, the associated $f_M\in \cF$ is a function representative of $\partial A_M$, the boundary notion introduced in Appendix \ref{app:missing_net}.

Moreover, for almost every $(v_s, v_b)\in A_M$, $q_M(v_s, v_b) = f_M(v_s)$. Similarly, letting $g_M$ be the generalized inverse of $f_M$, i.e. $g_M(v_b) = \inf\{x\in [0, 1]: f_M(x) \geq v_b\}$, we have that for almost every $(v_s, v_b)\in A_M$, $p_M(v_s, v_b) = g_M(v_b)$. We now state and prove the theorem.

\begin{theorem}
    If $v$ is a $\sigma$-smooth random variable on $[0, 1]^2$, then there exists a mechanism $M^\star\in \M$ such that the following holds: 
    $$
    \E{\prof(M^\star, v)} = \sup_{M\in \M} \E{\prof(M, v)}
    $$
\end{theorem}
\begin{proof}
    Let $\{f_n\}_n\subset \cF$ be a sequence such that $\E{\prof(f_n, v)} \xrightarrow{n} \sup_{f\in \cF}\E{\prof(f, v)}$. This sequence exists by definition of supremum. Since this is a uniformly bounded sequence of monotone increasing functions, we can apply Helly's Selection Theorem (see e.g. Chapter VIII, Section 4 in \citet{natanson2016theory}) to guarantee that $\exists \{f_{n'}\}_{n'}\subseteq \{f_n\}_n$ and $f^\star :[0, 1]\rightarrow [0, 1]$ such that $f_{n'}\rightarrow f^\star$ everywhere. Moreover, $f^\star$ is a monotone increasing function, which implies that it is almost everywhere continuous. It is not necessarily left continuous, but being monotone implies having left limits, so we can take left limits and enforce left continuity without consequences, since we are modifying the function in a set of null measure, and this operation leaves invariant the generalized inverse of $f^\star$, which we call $g^\star$. We therefore have $f^\star \in \cF$.
    
    Letting $\{g_{n'}\}_{n'}$ be the corresponding sequence of generalized inverses, we have that $g_{n'}(x)\rightarrow g^\star(x)$ for all $x \in [0, 1]$ where $g^\star$ is continuous. For a proof, see e.g. Proposition 0.1 in \citet{Resnick1987Extreme}. Since $f^\star$ is monotone, $g^\star$ is also monotone (due to $\{x\in[0,1]:f^\star(x)\geq y_2\}\subseteq \{x\in[0,1]:f^\star(x)\geq y_1\}$, for every $y_2\geq y_1$) and thus almost everywhere continuous: it follows that the previous convergence holds almost everywhere. 

    The indicator functions for the trade stemming from the sequence also eventually almost everywhere agree with the limit mechanism, due to $\sigma$-smoothness. Specifically, letting $A_{f^\star}$ be the epigraph of $f^\star$, we have that $\P{\cap_{m}\cup_{n'\geq m}(A_{f^\star}\Delta A_{f_{n'}})} = 0$ (where $\Delta$ stands for symmetric difference), meaning that $v \in A_{f^\star}\Delta A_{f_{n'}}$ infinitely often has probability zero. Therefore, it almost surely exists a $N$ such that $\ind{v\in A_{f^\star}} = \ind{v\in A_{f_{n'}}}\ \forall \, n' \geq N$. In turn, this implies that $\prof(f_{n'}, v)\rightarrow\prof(f^\star, v)$ almost surely, since we have proven that the two pricing functions converge almost everywhere and thus also almost surely due to $\sigma$-smoothness.

    At this point we can just use the Dominated Convergence Theorem ($\prof$ is bounded by $1$ in absolute value) to claim the following: 
    $$
    \E{\prof(f_{n'}, v)} \xrightarrow{n'} \E{\prof(f^\star, v)}.
    $$
    Recall that $\{f_{n'}\}_n \subseteq \{f_{n}\}_n$, and by definition of the latter, we thus have $\E{\prof(f^\star, v)} = \sup_{f\in \cF}\E{\prof(f, v)}$.

\end{proof}

\end{document}